\documentclass[journal,12pt,onecolumn,draftclsnofoot,]{IEEEtran}
\usepackage{amsmath,amsfonts}
\usepackage{algorithmic}
\usepackage{algorithm}
\usepackage{array}
\usepackage[caption=false,font=normalsize,labelfont=sf,textfont=sf]{subfig}
\usepackage{textcomp}
\usepackage{xcolor}
\usepackage{stfloats}
\usepackage{url}
\usepackage{verbatim}
\usepackage{graphicx}
\usepackage{cite}
\newtheorem{theorem}{Theorem}
\newtheorem{lemma}{Lemma}

\newtheorem{proof}{Proof}

\newtheorem{remark}{\textit{Remark}}
\newcommand{\bA}{\mathbf{A}}
\newcommand{\bB}{\mathbf{B}}
\newcommand{\bC}{\mathbf{C}}
\newcommand{\bD}{\mathbf{D}}
\newcommand{\bE}{\mathbf{E}}

\newcommand{\bI}{\mathbf{I}}

\newcommand{\bP}{\mathbf{P}}

\newcommand{\bR}{\mathbf{R}}

\newcommand{\bU}{\mathbf{U}}
\newcommand{\bV}{\mathbf{V}}
\newcommand{\bW}{\mathbf{W}}
\newcommand{\bX}{\mathbf{X}}
\newcommand{\bY}{\mathbf{Y}}
\newcommand{\bZ}{\mathbf{Z}}

\newcommand{\ba}{\mathbf{a}}
\newcommand{\bb}{\mathbf{b}}

\newcommand{\bu}{\mathbf{u}}
\newcommand{\bv}{\mathbf{v}}

\newcommand{\sB}{\mathcal{B}}
\newcommand{\sC}{\mathcal{C}}

\newcommand{\sF}{\mathcal{F}}
\newcommand{\sG}{\mathcal{G}}

\newcommand{\sK}{\mathcal{K}}
\newcommand{\sL}{\mathcal{L}}

\newcommand{\sP}{\mathcal{P}}
\newcommand{\sQ}{\mathcal{Q}}

\newcommand{\sS}{\mathcal{S}}
\newcommand{\sT}{\mathcal{T}}

\newcommand{\sY}{\mathcal{Y}}
\newcommand{\sZ}{\mathcal{Z}}

\newcommand{\bbB}{\mathbb{B}}

\newcommand{\bbP}{\mathbb{P}}

\newcommand{\bbR}{\mathbb{R}}

\newcommand{\sfA}{\mathsf{A}}
\newcommand{\sfB}{\mathsf{B}}
\newcommand{\sfC}{\mathsf{C}}

\newcommand{\sfE}{\mathsf{E}}

\newcommand{\sfX}{\mathsf{X}}
\newcommand{\sfY}{\mathsf{Y}}
\newcommand{\sfZ}{\mathsf{Z}}

\newcommand{\E}{\mathbb{E}}

\newcommand{\tp}{\text{T}}
\newcommand{\trace}{\text{trace}}
\newcommand{\rank}{\text{rank}}

\DeclareMathOperator{\vect}{vec}

\newcommand{\bbeta}{\boldsymbol{\beta}}
\newcommand{\balpha}{\boldsymbol{\alpha}}

\newcommand\numberthis{\addtocounter{equation}{1}\tag{\theequation}}

\begin{document}

\title{A tensor based varying-coefficient model for multi-modal neuroimaging data analysis}


\author{
    \IEEEauthorblockN{Pratim Guha Niyogi\textsuperscript{\textsection}\IEEEauthorrefmark{1}, Martin A. Lindquist\IEEEauthorrefmark{1}, Tapabrata Maiti\IEEEauthorrefmark{2}}
    \\
    \IEEEauthorblockA{\IEEEauthorrefmark{1}Department of Biostatistics, Johns Hopkins University, Baltimore, MD 21205}
    \IEEEauthorblockA{\IEEEauthorrefmark{2}Department of Statistics and Probability, Michigan State University, East Lansing, MI 48823
    }
}

\maketitle
\begingroup\renewcommand\thefootnote{\textsection}
\footnotetext{Corresponding author: pnyogi1@jhmi.edu}
\endgroup



\maketitle

\begin{abstract}
All neuroimaging modalities have their own strengths and limitations. A current trend is toward interdisciplinary approaches that use multiple imaging methods to overcome limitations of each method in isolation. At the same time neuroimaging data is increasingly being combined with other non-imaging modalities, such as behavioral and genetic data. The data structure of many of these modalities can be expressed as time-varying multidimensional arrays (tensors), collected at different time-points on multiple subjects. Here, we consider a new approach for the study of neural correlates in the presence of tensor-valued brain images and tensor-valued predictors, where both data types are collected over the same set of time points.  We propose a time-varying tensor regression model with an inherent structural composition of responses and covariates. Regression coefficients are expressed using the B-spline technique, and the basis function coefficients are estimated using CP-decomposition by minimizing a penalized loss function.  We develop a varying-coefficient model for the tensor-valued regression model, where both predictors and responses are modeled as tensors. This development is a non-trivial extension of function-on-function concurrent linear models for complex and large structural data where the inherent structures are preserved.  In addition to the methodological and theoretical development, the efficacy of the proposed method based on both simulated and real data analysis (e.g., the combination of eye-tracking data and functional magnetic resonance imaging (fMRI) data) is also discussed.
\end{abstract}

\par
\vspace{9pt}

\section{Introduction}
\IEEEPARstart{I}{n} recent years, there has been an explosive growth in the number of neuroimaging studies being performed. Popular imaging modalities include functional magnetic resonance imaging (fMRI), electroencephalography (EEG), diffusion tensor imaging (DTI),  positron emission tomography (PET), and single-photon emission-computed tomography (SPECT). Each of these techniques have their own limitations and strengths.  Therefore, a current trend is toward interdisciplinary approaches that use multiple imaging techniques to overcome limitations of each method in isolation. As an example, Figure \ref{fig:multimodal} illustrates the combination of fMRI and EEG data. At the same time, neuroimaging data is increasingly being combined with non-imaging modalities, such as behavioral and genetic data. Multi-modal analysis is an increasingly important topic of research, and to fully realize its promise, novel statistical techniques are needed. Here, we present a new approach towards performing such analysis.
\par
It is common for the data generated from neuroimaging studies to consist of time-varying signal measured over a large three-dimensional (3D) domain \cite{lindquist2008statistical, ombao2016handbook}.  Hence, the data are inherently spatio-temporal in nature.  Due to the massive size of the data along with its complex anatomical structure, classical vector-based spatio-temporal statistical methods are often deemed unrealistic and inadequate. It is becoming increasingly clear that any new model and methodology should address three fundamental concerns. First, standard spatio-temporal covariance modeling techniques are based on many parametric assumptions, which are often hard to validate in large high-dimensional data such as fMRI.  Second, modeling of spatio-temporal interactions often produces large covariance matrices containing millions of elements that are hard to estimate properly.  Third, storage of these large datasets while performing analysis is nearly impossible.
\par
The current research is motivated by the experiment \textit{studyforrest} (\url{http://studyforrest.org/}) which investigates high-level cognition in the human brain using complex natural stimulation, namely watching the Hollywood movie \textit{Forrest Gump} (1994).  The data consist of several hours of fMRI scans, structural brain images, eye-tracking data, and extensive annotations of the movie.   Details of this experiment are presented in Section \ref{sec:real-data}. In our motivating example, we focus on data consisting of voxel-wise fMRI images, measured over a large number of spatial locations (voxels) at 451 time-points. The goal of our analysis is to use the multivariate eye-tracking data, measured while the participants watch the movie, as covariates in a model that explains changes in the multivariate brain data. The vast size and scale of this data calls for well-equipped statistical techniques to find the association between brain regions and other covariates over time-varying activities.  It is useful to consider this as a regression problem with a multidimensional array of outcomes and predictors. These multidimensional arrays are popularly known as tensors.  Figure \ref{fig:whyTensor} illustrates the reason for considering a time-varying multidimensional array for the analysis.  Although the signals in both modalities (in this case fMRI and eye-tracking) are measured discretely over time, we consider them to be discrete measures of a smooth underlying function over time in a certain interval. This assumption is reasonable in the context of both brain activity and eye movement, as they can potentially change at any moment. 
\par
There are two main advantages to taking a tensor-based approach 
towards modeling this dataset. First, we can represent the unknown parameters to be estimated as a linear combination of rank-1 components, where the latter are expressed as the outer product of low-dimensional vectors.  This allows for the estimation of fewer parameters, which is consistent with variable selection or dimension reduction problems in statistics. Second, due to the need to estimate fewer parameters, the computational complexity is significantly reduced. 
\par
In a previous work, \cite{zhou2013tensor} formulated a regression framework that considers clinical outcomes as the response and images as covariates. Their method efficiently explored the spatial dependence of images in the form of a multi-dimensional array structure.  By extending the generalized linear regression to a multi-way parameter corresponding to the tensor-structured predictor, they proposed a penalized likelihood approach with adaptive lasso penalties, which are imposed on the individual margins of PARAFAC decomposition.  
A tensor-on-tensor regression approach was proposed in \cite{Lock2018}.
Furthermore, \cite{guhaniyogi2018bayesian} discussed a tensor response regression where the coefficients corresponding to each vector covariate are assumed to be tensors in the Bayesian framework.  Recently, \cite{liu2020low} have represented a generalized multi-linear tensor-on-tensor ridge regression model via tensor train representation.
\par
A varying-coefficient model in the functional data analysis (FDA) literature allows the regression coefficient to vary over some predictors of interest (say,  $T$).  In some cases, these predictors are confounded with covariates $\bX$ or some special variables, such as time.  This kind of model was first introduced and discussed by \cite{hastie1993varying} and has since been widely studied by researchers.  The non-constant relationship between functional response and predictors has been described in 
\cite{ramsay2005springer}.
\par
The current article provides the following contributions to this literature. First, we propose a method of modeling image data that can efficiently process large amounts of information and identify associations while preserving the structure of the 3D images and multi-layer covariates. Second, we consider the time-varying function-on-function concurrent linear model \cite{hastie1993varying} and generalize it to the tensor-on-tensor regression case, thus moving a step further than \cite{Lock2018}, which did not consider the time-varying coefficient.  Consequently, our generalization provides an extension to classical functional concurrent regression with tensor predictors and tensor covariates.  To the best of our knowledge, such an approach has not yet been proposed in statistics literature. Here, we express the regression coefficients using the B-spline technique, and the coefficients of the basis functions are estimated using CP-decomposition, thereby reducing computational complexity.  Furthermore, our model requires minimum assumptions compared to those in the existing literature. Our approach does not require the estimation of covariance separately. Thus, our proposal offers an important addition to the literature on functional and imaging data analysis. Our methods are flexible and general; therefore, they are applicable using data from different domains such as multi-phenotype analysis and imaging genetics.  This makes it an ideal approach for modeling multi-modal data of the type described in our motivating example.
\par
The rest of the article is organized as follows. Section \ref{sec:notation} reviews the notation and properties of the matrix and array.  The proposed tensor-on-tensor functional regression models are described in Section \ref{sec:general}.  Section \ref{sec:main-result} provides the theoretical properties of the proposed estimator.  Section \ref{sec:implementation} presents the algorithm and implementation of the method.  The simulation results are presented in Section \ref{sec:simulation} and real data examples are shown in Section \ref{sec:real-data}. 
Section \ref{sec:discussion} concludes with a discussion of future extensions. Technical proofs are presented in the appendix.

\section{Basic notations, definitions and properties}
\label{sec:notation}
In this section, multi-dimensional arrays, also known as tensors, play an important role. 
We begin with a brief summary of tensors for completeness purpose and define important notation which will be utilized in the rest of the paper. Interested readers can refer to a survey article by \cite{Kolda2009} for more information. 
\par
Throughout this paper, we denote tensors using Sans-serif upper-face letters
$(\sfA, \sfB, \cdots)$,
matrices using bold-face capital letters
$(\bA, \bB \cdots)$, 
vectors using bold-face lower-case letters $(\ba, \bb, \cdots)$, and
scalars as non-bold lower-case letters $(a, b, \cdots)$. 
The entry in the $i$-th row and $j$-th column of a matrix $\bA$ is denoted as $ (\bA) _{i, j} = a_{ij}$ 
and the $(i_{1}, \cdots, i_{D})$-th entry of a $D$ dimensional tensor is denoted as $(\sfA)_{i_{1}, \cdots, i_{D}} = a_{i_{1}, \cdots, i_{D}}$. 
For a $D$-way tensor $\sfA \in \bbR^{I_{1} \times \cdots \times I_{D}}$ with element $a_{i_{1}, \cdots, i_{D}}$ at position with mode $i_{d}, d = 1, \cdots, D$, 
vectorization operator $\vect(\cdot)$ is defined as a vector of length $\prod_{d = 1}^{D}I_{d}$ where 
	$\vect(\sfA)\left[ i_{1} + \sum_{d = 2}^{D} \left( \prod_{k = 1}^{d-1}I_{k}\right)  (i_{d} - 1) \right]  = a_{i_{1}, \cdots, i_{D}}$.
Similarly, one can perform $d$-mode matricization, or unfolding, on a $D$-array $\sfA$, to form a matrix $\bA_{(d)}$ with $I_{d}$ rows 
and $\prod_{d': d' \neq d}I_{d'}$ columns where the element $a_{i_{1}, \cdots, i_{D}}$ is at the row $i_{d}$ and column 
$\left\{1+\sum_{d_{1}=1 (\neq d)}^{D}(i_{d_{1}}-1)\prod_{d_{2}=1(\neq d)}^{d_{1}-1}I_{d_{2}}\right\}$, which reshapes the tensor to a matrix corresponding to a fixed mode.
A $D$-way tensor $\sfA$ has rank-1 when it is the outer product of the $D$ vectors $\bu^{(1)}, \cdots, \bu^{(D)}$ 
which is denoted by $\bu^{(1)}\circ \cdots\circ\bu^{(D)}$. 
Mathematically, $a_{i_{1}, \cdots, i_{D}} = u_{i_{1}}^{(1)}u_{i_{2}}^{(2)}\cdots u_{i_{D}}^{(D)}$ for all possible choices of indices $(i_{1}, \cdots, i_{D})$. 
The rank of a tensor $\sfA$ is $R$ if it is the minimal number of rank-1 tensors that form $\sfA$ as a linear combination. \par
Now the question is how to express the tensor as the sum of a finite number of rank-one tensors? The answer comes from Psychometrics in the form of canonical decomposition or CANDECOMP 
and parallel factors or PARAFAC 
and from the literature on tensor decomposition where CANDECOMP/PARAFAC (CP) decomposition provides an extension of matrix singular value decomposition \cite{kiers2000towards, tucker1966some}. 
CP decomposition, therefore, factorizes a tensor into a sum of component rank-one tensors, mathematically, 
$\sfA = \sum_{r = 1}^{R} \bu_{r}^{(1)}\circ\cdots\circ\bu_{r}^{(D)}$
where  $\bu_{r_{d}}^{(d)} \in \bbR^{I_{d}}, d = 1, \cdots, D$ are column vectors and $\sfA$ cannot be written as a sum of less then $R$ outer product.
Often CP decomposition is indicated by $\sfA = \left[ \left[ \bU_{1}, \cdots, \bU_{D} \right] \right] $ where $\bU_{1}, \cdots, \bU_{D}$ have linearly independent columns $\bU_{d} = [\bu_{1}^{(d)}, \cdots, \bu_{R}^{(d)}] \in \bbR^{I_{d} \times R}$ for each $d = 1, \cdots D$. 
There are several kinds of inner products for higher-order tensors. The scalar product $\left\langle \sfA, \sfB \right\rangle$ of two $D$-dimensional tensors is defined as
$\left\langle \sfA, \sfB \right\rangle = \sum_{i_{1}, \cdots, i_{D}} b_{i_{1}, \cdots, i_{D}}a_{i_{1}, \cdots, i_{D}}$
The Frobenius norm of tensor $\sfA$ is defined as $\left\| \sfA \right\|_{\sF} = \sqrt{\left\langle \sfA, \sfA \right\rangle}$. 
In this paper, we consider the contracted tensor product between two tensors with different mode dimensions. 
For two tensors $\sfA \in \bbR^{I_{1}\times \cdots \times I_{K}\times P_{1} \times \cdots \times P_{L}}$ and 
$\sfB \in \bbR^{P_{1} \times \cdots \times P_{L}\times Q_{1}\times \cdots \times Q_{M}}$, contracted tensor product \cite{Lock2018, raskutti2019} is defined as $\left\langle  \sfA, \sfB \right\rangle_{L}$ with $(i_{1}, \cdots, i_{K}, q_{1}, \cdots, q_{M})$-th element $\sum_{p_{1}, \cdots, p_{L}}a_{i_{1}, \cdots, i_{K}, p_{1}, \cdots, p_{L}} b_{p_{1}, \cdots, p_{L}, q_{1}, \cdots, q_{K}}$. 
\section{Tensor-on-tensor functional regression}
\label{sec:general}
In this section, we discuss tensor-on-tensor functional regression with time-varying coefficients. 
Let $\sfY(t) \in \bbR^{Q_{1}\times\cdots\times Q_{M}}$ with $(q_{1}, \cdots, q_{M})$-th element $y_{q_{1}, \cdots, q_{M}}$ for all possible indices be a set of time-varying response variables observed at time $t$ and $\left\lbrace \sfY(t) : t\in\sT\right\rbrace$ be the underlying continuous stochastic process defined on a compact interval $\sT$. 
Without loss of generality, we assume $\sT = \left[ 0, T\right], T > 0$. Suppose there are $N$ individuals/trajectories on $\sT$. 
Observations are taken at $J$ distinct points for each individual. 
Collection of points for the $i$-th individual is denoted as $\sT^{\dagger}_{i} = \left\lbrace 0 \leq t_{i1} < \cdots < t_{iJ} \leq T \right\rbrace$. 
Therefore, for $i$-th individual at a set of discrete time-points $\sT^{\dagger}_{i}$, we observe the responses $\sfY_{i}(t_{i}) = \left( \sfY_{i}(t_{i1}), \cdots, \sfY_{i}(t_{iJ}) \right) \in \bbR^{J \times Q_{1}\times\cdots\times Q_{M}}$ which are distinct realizations of the corresponding stochastic process.
The covariate $\sfX(t) \in \bbR^{P_{1}\times\cdots\times P_{L}}$ with $(p_{1}, \cdots, p_{L})$-th element $x_{p_{1}, \cdots, p_{L}}(t)$ for all indices, observed at $\sT^{\dagger}_{i}$ is denoted as $\sfX_{i}(t_{i}) = ( \sfX_{i}(t_{i1}), \cdots, \sfX_{i}(t_{iJ})) \in \bbR^{J \times P_{1}\times\cdots\times P_{L}}$. 
The time-varying tensor coefficient $\bbeta(t) \in \bbR^{P_{1}\times \cdots\times P_{L} \times Q_{1} \times \cdots \times Q_{M}}$ is assumed to vary over time smoothly. 
Therefore, we can apply local polynomial smoothing\cite{eubank1999nonparametric},
smoothing splines \cite{green1993nonparametric},
regression splines\cite{eubank1999nonparametric},
P-splines \cite{ruppert2003semiparametric}. 
 In this paper, we use B-spline bases which are very popular in mathematics, computer science, and statistics \cite{de1978practical}.
Now, for $1 \leq p_{l} \leq P_{l}, 1 \leq q_{m} \leq Q_{m}, 1\leq l \leq L, 1\leq m \leq M$, each function $\beta_{p_{1}, \cdots, p_{L} , q_{1}, \cdots, q_{M}}(t)$ can be approximated by 
\begin{align*}
\label{eq:basisExp}    
&\beta_{p_{1}, \cdots, p_{L} , q_{1}, \cdots, q_{M}}(t) \\
&= \sum_{h = 1}^{H}b_{h, p_{1}, \cdots, p_{L} , q_{1}, \cdots, q_{M}}\bbB_{h}(t) = \bb_{p_{1}, \cdots, p_{L} , q_{1}, \cdots, q_{M}}^{\tp}\sB(t)
\numberthis
\end{align*}
where $\bb_{p_{1}, \cdots, p_{L} , q_{1}, \cdots, q_{M}} = (b_{1, p_{1}, \cdots, p_{L} , q_{1}, \cdots, q_{M}}, \cdots, b_{H, p_{1}, \cdots, p_{L} , q_{1}, \cdots, q_{M}})^{\tp}$ is the collection of basis coefficients and $\sB(t) = (\bbB_{1}(t), \cdots, \bbB_{H}(t))^{\tp}$ is a vector of known B-spline bases. 
\par
In practice, we can use mode-wise different basis functions to approximate $\beta_{p_{1}, \cdots, p_{L} , q_{1}, \cdots, q_{M}}(t)$.
However, for convenience, we use the same set of bases in this paper.
Instead of B-spline, one can use other basis functions to approximate the coefficient functions. 
We use the B-spline base for its simplicity and numerical tractability. Although this method does not produce a desirable approximation for discontinuous functions, in this paper,
we restrict ourselves to smooth continuous coefficients.\par
We propose a general time-varying tensor-on-tensor regression model,
\begin{equation}
    \sfY_{i}(t) = \left<\sfX_{i}(t), \bbeta(t)\right>_{L}+\sfE_{i}(t)
\end{equation}
which can be reduced into the following mode-wise time-varying coefficient model.
\begin{align*}
&y_{i, q_{1}, \cdots, q_{M}} (t)\\
&= \sum_{p_{1} = 1}^{P_{1}}\cdots \sum_{p_{L} = 1}^{P_{L}} x_{i, p_{1}, \cdots, p_{L}} (t) \beta_{p_{1}, \cdots, p_{L} , q_{1}, \cdots, q_{M}}(t) + \epsilon_{i, q_{1}, \cdots, q_{M}}(t)
\numberthis
\end{align*}
where $\epsilon_{i, q_{1}, \cdots, q_{M}}(t)$ is a random error with mean zero. Errors can be correlated over time and modes, but are independent over the trajectories. After plugging-in the approximate expression of $\beta_{p_{1}, \cdots, p_{L} , q_{1}, \cdots, q_{M}}(t)$ at each mode, the model can now be expressed as
\begingroup
\allowdisplaybreaks
\begin{align*}
\label{eq:model}
&y_{i, q_{1}, \cdots, q_{M}} (t)\\
&= \sum_{p_{1} = 1}^{P_{1}}\cdots \sum_{p_{L} = 1}^{P_{L}} \sum_{h = 1}^{H}b_{h, p_{1}, \cdots, p_{L} , q_{1}, \cdots, q_{M}} x_{i, p_{1}, \cdots, p_{L}} (t) \bbB_{h}(t)\\
&\qquad + \epsilon_{i, q_{1}, \cdots, q_{M}}(t)  
\numberthis
\end{align*}
\endgroup
The multi-dimensional basis coefficients $\sfB_{0}= \left\lbrace b_{h, p_{1}, \cdots, p_{L} , q_{1}, \cdots, q_{M}}:\right.$ $1 \leq h \leq H$, $1 \leq p_{l} \leq P_{l}$, $1\leq q_{m} \leq Q_{m}$, $1 \leq l \leq L$, $\left.1\leq m \leq M\right\rbrace$ 
can be estimated by minimizing mode-wise penalized integrated sum of square errors with respect to $\sfB_{0}$. 
Let us denote the smoothness penalty by 
$\Omega_{sm}$ where
\begingroup
\allowdisplaybreaks
\begin{align*}
\label{eq:SmoothPen}    
\Omega_{sm}(\sfB_{0})
&=\sum_{p_{1} = 1}^{P_{1}}\cdots \sum_{p_{L} = 1}^{P_{L}}\sum_{q_{1} = 1}^{Q_{1}} \cdots \sum_{q_{M}= 1}^{Q_{M}}
\int \theta_{p_{1}, \cdots, p_{L}, q_{1}, \cdots, q_{M}} \\
& \qquad \times 
\left\lbrace  \beta_{p_{1}, \cdots, p_{L} , q_{1}, \cdots, q_{M}}''(t) \right\rbrace^{2} dt \\
&= \sum_{p_{1} = 1}^{P_{1}} \cdots \sum_{p_{L} = 1}^{P_{L}}\sum_{q_{1} = 1}^{Q_{1}} \cdots \sum_{q_{M} = 1}^{Q_{M}}  \theta_{p_{1}, \cdots, p_{L}, q_{1}, \cdots, q_{M}} \\
& \qquad \times \bb_{p_{1}, \cdots, p_{L} , q_{1}, \cdots, q_{M}}^{\tp}\int \bB''(t)\bB''(t)^{\tp}dt  \bb_{p_{1}, \cdots, p_{L} , q_{1}, \cdots, q_{M}}
\numberthis
\end{align*}
\endgroup
Hence, the loss function turns out to be
\begin{align*}
    \label{LossFunction}    
&\sL(\sfB_{0})\\
&=  \frac{1}{N}\int_{\sT}\sum_{i = 1}^{N} \sum_{q_{1}=1}^{Q_{1}}\cdots \sum_{q_{M} = 1}^{Q_{M}} \bigg(  y_{i, q_{1}, \cdots, q_{M}}(t) \bigg. \\
& \bigg.  
-   \sum_{p_{1} = 1}^{P_{1}}\cdots \sum_{p_{L} = 1}^{P_{L}}\sum_{h = 1}^{H}b_{h, p_{1}, \cdots, p_{L} , q_{1}, \cdots, q_{M}} x_{i, p_{1}, \cdots, p_{L}} (t) B_{h}(t)\bigg)^{2}dt  \\ 
&\qquad + \Omega_{sm}(\sfB_{0})  
\numberthis    
\end{align*}
In Equation (\ref{eq:SmoothPen}), 
$\lbrace \theta_{p_{1}, \cdots, p_{L}, q_{1}, \cdots, q_{M}}\rbrace_{p_{1}, \cdots, p_{L}, q_{1}, \cdots, q_{M}}$ are the tuning parameters for smoothness. 
The use of smoothness penalties is widespread in the functional data analysis literature (see \cite{ramsay2005springer} among many others). 
In practice, it is unrealistic to determine these large numbers of pre-assigned tuning parameters.
By considering $\theta_{p_{1}, \cdots, p_{L}, q_{1}, \cdots, q_{M}} = \theta$, for all possible $p_{1}, \cdots, p_{L}, q_{1}, \cdots, q_{M}$, the simplest version of smoothness penalty would be, $\Omega_{sm}(\sfB_{0}) = \theta \vect(\sfB_{0})^{\tp}(\bI_{Q} \otimes \bI_{P} \otimes \int \bB''(t)\bB''(t)^{\tp}dt )\vect(\sfB_{0})$.
Note $\vect(\sfB_{0}) = (\bb_{11}, \cdots, \bb_{P1}, \bb_{12}, \cdots, \bb_{P2}, \cdots, \bb_{1Q}, \cdots, \bb_{PQ})^{\tp}$.
Therefore, the penalized likelihood estimating equation for the functional tensor-on-tensor regression problem is 
\begin{equation}
\label{LossFunctionTensor}
    \sL(\sfB_{0}) = \int_{\sT}\frac{1}{N}\sum_{i = 1}^{N} \left\| \sfY_{i}(t) - \left\langle \sfZ_{i}(t), \sfB_{0}\right\rangle_{L+1} \right\|_{\sF}^{2}dt + \Omega_{sm}(\sfB_{0})
\end{equation}
where $\left\langle \cdot, \cdot\right\rangle_{L+1}$ is the contracted tensor product defined in Section \ref{sec:notation} and $\|\cdot\|_{\sF}$ is the Frobenius norm. 
The first term of Equation (\ref{LossFunctionTensor}) is the integrated sum of squares, and the second term is the smoothness penalty.\par
Let the response tensor for time $t$, $\sfY(t) \in  \bbR^{N \times Q_{1}\times \cdots\times Q_{M}}$ with its $(i, q_{1}, \cdots, q_{M})$-th element be $y_{i, q_{1}, \cdots, q_{M}}(t)$ for all $i = 1, \cdots, N; \, q_{m} = 1, \cdots, Q_{m};  \, m = 1, \cdots, M$.
Similarly, we define an updated covariate tensor contaminated with B-spline bases $\sfZ(t) \in \bbR^{N \times H \times P_{1}\times \cdots\times P_{L}}$ where the $(i, h, p_{1}, \cdots, p_{L})$-th element of the tensor is defined as $z_{i,h, p_{1}, \cdots, p_{L}}(t) = x_{i, p_{1}, \cdots, p_{L}}(t)\bbB_{h}(t)$. Therefore, the corresponding penalized loss function in Equation (\ref{LossFunctionTensor}) is equivalent to $\sL(\sfB_{0}) = \int_{\sT}\left\|\sfY(t) - \left\langle \sfZ(t), \sfB_{0}\right\rangle_{L+1} \right\|_{\sF}^{2}dt + \Omega_{sm}(\sfB_{0})$.
\begin{remark}
For $Q = 0$, the proposed model reduces to the classical concurrent linear model \cite{ramsay2005springer}. 
For $Q = 1$ and $P = 1$, the time-varying network model \cite{Xue2018} is a special case of our proposed model for a specific choice of covariates. 
For $Q = 2$, $y_{i, q_{1}, q_{2}}(t)$ is the observation of the quantity of interest at time $t$ for sub-unit $q_{2}$ from unit $q_{1}$ of a treatment group $i$ in a hierarchical model \cite{zhou2010reduced}. 
\end{remark}
\par
Let $P = \prod_{l = 1}^{L}P_{l}$ be the total number of predictors for each observation and $Q = \prod_{m = 1}^{M}Q_{m}$ be the total number of outcomes for each predictor over time. 
To minimize the penalized integrated sum of squared residuals described, the solution for $\sfB_{0}$ might be inconsistent. Since the unknown coefficient tensor $\sfB_{0}$ has $H\prod_{l = 1}^{L} P_{l}  \prod_{m = 1}^{M} Q_{m}$ parameters, we need to adopt a dimension reduction technique. 
Inspired by the novel idea discussed in \cite{Lock2018}, 
we consider the rank $R$ decomposition of $\sfB_{0}$ as $\sfB_{0} = \left[ \left[ \bU_{0}, \bU_{1}, \cdots, \bU_{L} , \bV_{1}, \cdots, \bV_{M} \right] \right]$ where $\bU_{0}$, $\bU_{l}$ and $\bV_{m}$ are matrices with dimensions  $H\times R$, $P_{l}\times R$ and  $Q_{m}\times R$, respectively, for all $1 \leq l \leq L, 1 \leq m \leq M$. 
After dimension reduction, the number of unknown parameters reduces to $R(H+\sum_{l = 1}^{L} P_{l}+\sum_{m = 1}^{M} Q_{m})$. 
Therefore, the estimate of the coefficient tensor is $\tilde{\sfB}_{0} = \arg\min_{\rank(\sfB_{0}) \leq R} \sL(\sfB_{0})$.
However, this estimated coefficient tensor suffers from over-fitting and instability problems due to multi-collinearity of $\sfZ$ and/or the large number of observed outcomes. 
Thus, we obtain an alternative estimate of coefficient tensor $\sfB_{0}$ as $\hat{\sfB}_{0} = \arg\min_{\rank(\sfB_{0}) \leq R}  \sQ(\sfB_{0})$
based on the modified loss function,  $\sQ$, defined by 
\begingroup
\allowdisplaybreaks
\begin{equation}
\label{LossMod}
 \sQ(\bB_{0}) =  \frac{1}{N}\int_{\sT}\left\| \sfY(t) - \left\langle \sfZ(t), \sfB_{0}\right\rangle_{L+1} \right\|_{\sF}^{2}dt + \Omega(\sfB_{0})
\end{equation}
\endgroup
where 
\begingroup
\allowdisplaybreaks
\begin{align*}
\label{pen}
&\Omega(\sfB_{0})\\
&= \theta \vect(\sfB_{0})^{\tp} (\bI_{Q} \otimes \bI_{P} \otimes \int \bB''(t)\bB''(t)^{\tp}dt )\vect(\sfB_{0})\\
& \qquad + \phi \vect(\sfB_{0})^{\tp} \vect(\sfB_{0}).
\numberthis
\end{align*}
\endgroup
Equation (\ref{pen}) suggests performing penalization of the smoothness and sparsity of the coefficient functions simultaneously. 
\begin{remark}
For fixed rank $R_{0}$, the number of knots and tuning parameters $\theta$ and $\phi$ are unknown and can be selected using Mallows's $C_{p}$ \cite{mallows151973}, generalized cross-validation  \cite{craven1978smoothing}.
To choose the rank of the CP-decomposition, we choose BIC-type information criterion for chosen tuning parameters $\widehat{\theta}$ and $\widehat{\phi}$, $\text{BIC} = -2l(\sfB_{0}(R_{0}); \widehat{\theta}, \widehat{\phi}) + \log(NJ)p_{e}$, where $l$ is the log-likelihood evaluated at $\sfB$ with working rank $R_{0}$ and $p_{e}$ is the effective number of parameters. 
\end{remark}

\section{Theory}
\label{sec:main-result}
In this section, we will study identifiablity of the model and consistency of the parameter estimates under our proposed model as the number of subjects $N$ goes to infinity, while assuming that the rank of the basis tensor coefficient is known and fixed.
\subsection{Identifiability}
\label{subsec:identifiability}
Identifiability issues play important roles in tensor regression \cite{Lock2018,zhou2013tensor,guhaniyogi2017bayesian}. The model discussed in Section \ref{sec:general} would be identifiable for $\bbeta(t)$, if $\bbeta(t) \neq \bbeta^{*}(t)$ implies $\left< \sfX(t), \bbeta(t) \right>_{L} \neq \left< \sfX(t), \bbeta^{*}(t) \right>_{L}$ for some $t \in \sT$ and some $\sfX(t) \in \bbR^{P_{1} \times \cdots \times P_{L}}$. 
Using the basis expansion in Equation (\ref{eq:basisExp}), we can say that $\sfB_{0}$ is identifiable if and only if $\bbeta(t)$ is identifiable for all $t \in \sT$. 
Therefore, the reduced model is identifiable if $\sfB_{0} \neq \sfB_{0}^{*}$ implies $\left< \sfZ(t) , \sfB_{0}\right>_{L+1} = \left< \sfZ(t) , \sfB_{0}^{*}\right>_{L+1}$ for some $t \in \sT$ and for some $\sfZ(t) \in \bbR^{H \times P_{1} \times \cdots \times P_{L}}$.
Let us assume, for $t = t_{0}$,  $\sfZ_{h, p_{k_{1}}, \cdots, p_{k_{L}}}(t_{0}) = 1$ at $k_{1} = 1, \cdots, k_{L} = L$ and 0 otherwise, then the product becomes $b_{h, p_{1}, \cdots, p_{L}, q_{1}, \cdots, q_{M}}$. 
Furthermore, $\bU_{0}, \bU_{1}, \cdots, \bU_{l}, \bV_{1}, \cdots, \bV_{M}$ in the expression of CP-decomposition is not identifiable. 
Therefore,
the identifiability conditions can be imposed in the following way \cite{sidiropoulos2000uniqueness}.
\begin{enumerate}
    \item Restrictions for scale and non-uniqueness: $\sfB_{0}$ will remain the same after replacing $\bU_{0}$, $\bU_{l}$ and $\bV_{m}$ by $c_{s}\bU_{0}$, $c_{u_{l}}\bU_{l}$ and $c_{v_{m}}\bV_{m}$ respectively, 
    where $\lbrace c_{s},c_{u_{l}}, c_{v_{m}}\rbrace$ is the set of constants with $c_{s}\prod_{l = 1}^{L}c_{u_{l}}\prod_{l = 1}^{M}c_{v_{l}} = 1$. 
    This problem can be solved by introducing the condition that the norm of each of $\bu_{r0}$, $\bu_{rl}$ and $\bv_{rm}$ is set to 1, $1 \leq r \leq R, 1 \leq l \leq L, 1 \leq m \leq M$. 
    \item Restriction for permutation: For any permutation $\pi(\cdot)$ of $\lbrace 1, \cdots, R \rbrace$, $\sum_{r = 1}^{R} \bu_{r0}\circ \bu_{r1} \circ \cdots \circ \bu_{rL} \circ \bv_{r1} \circ \cdots \circ \bv_{rM}$ 
    is the same as 
    $\sum_{r = 1}^{R} \bu_{\pi(r)0}\circ \bu_{\pi(r)1} \circ \cdots \circ \bu_{\pi(r)L} \circ \bv_{\pi(r)1} \circ \cdots \circ \bv_{\pi(r)M}$. 
    Therefore, we impose the restriction $\|\bu_{01}\| \geq \cdots \geq \|\bu_{0R}\|$.
\end{enumerate}
These conditions ensure identifiability for $L+M \geq 2$. Therefore, we do not need the additional orthogonality condition used in \cite{Lock2018,zhou2013tensor,guhaniyogi2017bayesian}. 
\subsection{Convergence rate}
\label{subsec:rate}
In this subsection, we study the asymptotic properties of the estimate of time-varying tensor regression parameter $\bbeta(t)$ based on polynomial spline approximation and the CP decomposition. 
Since the number of modes is fixed, we reduce the objective function following the notation $\bY \in \bbR^{NJ \times Q}$ and $\sfZ \in \bbR^{NJ\times H_{N} \times P}$ and, therefore, $\sQ(\bB_{0}) = \frac{1}{NJ}\|\bY - \left<\sfZ, \sfB_{0} \right>_{2}\|_{\sF}^{2} 
    + \| \sfB_{0}\|_{\sF, \bW_{\omega}}^{2}$
where $\| \sfB_{0}\|_{\sF, \bW_{\omega}}$ is the weighted Frobenius norm defined as $\| \sfB_{0}\|_{\sF, \bW_{\omega}} = \sqrt{\vect(\sfB_{0})^{\tp}\bW_{\omega}\vect(\bB_{0})}$ where $\omega$ is a set of tuning parameters. 
Moreover, assume that $\rank(\sfB_{0}) = R_{0}$ which is assumed to be known and fixed. 
To proceed further, we introduce some regularity conditions required to establish the asymptotic properties. 
\renewcommand{\labelenumi}{(C\arabic{enumi})}
\begin{enumerate}
    \item \label{cond:time} The observation times $t_{ij}$ for $i = 1, \cdots, N; j = 1, \cdots, J$ are independent and follow a distribution $f_{T}(t)$ over the support $\sT$. The density function $f_{T}(t)$ is assumed to be absolutely continuous and bounded by a nonzero and finite constant.
    \item \label{cond:knots} Let $\left\lbrace \tau_{h}\right\rbrace_{h = 1}^{K_{n}}$ be $K_{n}$ interior knots within the compact interval $\sK = \left[0, 1\right]$  and denote the partition of the interval $\left[0, T \right]$  with $K_{N}$ knots as $\mathbb{I} = \left\lbrace 0 = \tau_{0} < \tau_{1} < \cdots < \tau_{K_{N}} < \tau_{K_{N+1}} = 1 \right\rbrace$. 
    \item \label{cond:spline} The polynomial spline of order $v + 1$ are the functions with degree $v$ of polynomials on the interval $\left[ \tau_{h - 1}, \tau_{h} \right)$ for $h = 1, \cdots, K$ and $\left[\tau_{K_{N}}, \tau_{K _{N}+ 1} \right]$ and $v - 1$ continuous derivatives globally. 
    \item \label{cond:ep} For $t \in \sT$, $\epsilon_{i, q_{1}, \cdots, q_{M}}(t)$'s are i.i.d. copies with mean zero and finite second order moment over $i$. Moreover, for each $i$ the coordinates $q_{1}, \cdots, q_{M}$, $\epsilon_{i, q_{1},\cdots, q_{M}}(t_{ij})$ are locally stationary time series of the form given in appendix. 
    Assume the physical dependence measure $\Delta(k, a)$ is upper bounded by $k^{-\kappa_{0}}$ for some positive $\kappa_{0}$ and for all $j \geq 1$. 
    \item \label{cond:x} The covariates $x_{i, p_{1}, \cdots, p_{L}}(t)$ are i.i.d. for index $i$ and they is bounded almost everywhere.
    \item \label{cond:eigen}
   $\lambda_{\min}\left(\sfZ_{(1)}^{\tp}\bZ_{(1)}\right) 
    = \sigma_{\min}(\sfZ_{(1)})^{2} \geq \lambda_{\min}(\bbB^{\tp}\bbB)\lambda_{\min}(\bX^{\tp}\bX) > \lambda$ where $\lambda_{i}(\bA)$ and $\sigma_{i}(\bA)$ denote $i$-th eigen-value and singular value, respectively, for a matrix $\bA$. 
\end{enumerate}
\begin{remark}
Conditions \ref{cond:time}, \ref{cond:knots}, \ref{cond:spline} are standard conditions in the context of polynomial spline regression and are required to ensure the consistency of the spline estimation of the varying-coefficient models. Condition \ref{cond:spline} provides the degree of smoothness on the time-varying coefficients. We assume condition \ref{cond:ep} to represent a wide class of stationary, locally stationary, and non-linear processes. Similar conditions can be found in \cite{ding2020estimation, ding2021multivariate}. This is a natural assumption of temporal short-range dependency where temporal correlation decays in polynomial order. This phenomenon can also be observed in the well-known Ornstein–Uhlenbeck process and the linear process with the standard basis expansion $\epsilon_{i, \bullet}(t) = \sum_{k = 1}^{\infty}a_{ik, \bullet}\phi_{k}(t)$ where $a_{ik, \bullet}$ is an uncorrelated mean zero, finite variance random variable over $(i, k)$ and $\sup_{t}\phi_{k}(t) \leq Ck^{-a}$ for some positive constants $C$ and $a$ . 
\end{remark}
Define, the constants $\sC(\delta) = 1 + 2/\delta$ such that $\sC(\delta)\leq \lambda^{2}/2\mu$ where $\mu = (NJ)(\theta\lambda_{\max}(\int \bB''(t)\bB''(t)^{\tp}dt) + \phi)\sqrt{2R_{0}}$. 
Further define, $\xi = \sup_{1\leq h \leq H}\sup_{t\in [0, 1]}|\bbB_{h}(t)|$ which is typically bounded. Further, define $\sigma_{1}(\sfC) = \max\{\sigma_{1}(\sfC_{(1)}), \sigma_{1}(\sfC_{(2)}), \sigma_{1}(\sfC_{(3)}) \}$.
We propose the following theorem for the estimation and prediction performance of the coefficient tensor. 
\begin{theorem}
\label{theorem:coeffTensor}
Under assumptions \ref{cond:ep} and \ref{cond:eigen}, when both the number of time-points and trajectories are large enough, there exists a constant $C_{a}$, so that
with probability at least 
$1-C_{a}N^{-a\tau}$,
\begin{align*}
 &\|\left< \sfZ, (\hat{\sfB}_{0} - \sfB_{0})\right>\|_{\sF}^{2}\\
        & \leq \lambda^{-1}\left(
            \sC(\delta)^{-1} - 2\mu\lambda^{-2}
        \right)^{-1}\\
        & \qquad \times
        \left\{
            4\mu\sigma_{1}^{2}(\sfC) + {2R_{0}}(1+\delta)Q^{2}\xi^{2} N^{2\tau+2}{J}
        \right\}   
    \numberthis
\end{align*}
for any $H_{N} \times P \times Q$ matrix $\sfC$ with $\rank(\sfC) \leq R_{0}$, By choosing $\sfC = \sfB_{0}$, a simplified prediction error could be obtained. Under the same set of assumptions, the estimation error of the matrix $\sfB_{0}$ is 
$\|\hat{\bB}_{0} - \bB_{0}\|_{\sF}^{2} \leq \lambda^{-1}\left(
            \sC(\delta)^{-1} - 2\mu\lambda^{-2}
        \right)^{-1}
        \left\{
            4\mu\sigma_{1}^{2}(\sfC) + {2R_{0}}(1+\delta)Q^{2}\xi^{2} N^{2\tau+2}{J}
        \right\}$.
\end{theorem}
Additionally, we introduce the following theorem, which states the consistency result for the coefficient tensor function.
\begin{theorem}
\label{theorem:coeffFunctionTensor}
Under assumptions \ref{cond:time}-\ref{cond:eigen}, with probability,  we have the following with probability $1-C_{a}N^{-a\tau}$,
\begin{align*}
    &\int_{\sT}|\hat{\beta}_{\bullet}(t) - \beta_{\bullet}(t)|^{2}f_{T}(t)dt \\
    &= O\left\{
    \lambda^{-1}\left(
            \sC(\delta)^{-1} - 2\mu\lambda^{-2}
        \right)^{-1}\right.\\
    &  \times \left.\left\{4\mu\sigma_{1}^{2}(\sfC_{(1)}) + {2R_{0}}(1+\delta)Q\xi N^{\tau+1}\sqrt{J}\right\}\right.
         \left.  + 
        K_{N}^{-2(v+1)}
    \right\}
    \numberthis
\end{align*}

\end{theorem}


\section{Algorithm and implementation}
\label{sec:implementation}
In this section, we propose a general algorithm to estimate the basis coefficient tensor using the objective function described in Section \ref{sec:general}. 
For given time-points $t_{1}, \cdots, t_{J}$, define $\sZ$ and $\sY$ as the combined tensor after staking over all time-points.
Therefore, $\sZ$ and $\sY$ are the tensors of order $NJ \times H \times P_{1}\times \cdots P_{L}\times Q_{1}\times\cdots Q_{M}$ and $NJ \times Q_{1}\times \cdots Q_{M}$, respectively. 
Moreover define, $\breve{\bB}_{0}$ as the matrix of coefficient of order $HP\times Q$, where columns and rows of $\sfB_{0}$ are obtained by vectorizing first $(L+1)$ and last $M$  modes of $\sfB_{0}$, respectively. 
For the alternate expression of the penalty term in Equation (\ref{LossMod}), observe 
(i) $\left\| \sfB_{0} \right\|^{2} = \| \breve{\bB_{0}}\|^{2} = \vect(\sfB_{0})^{\tp}\vect(\sfB_{0}) = \trace(\breve{\bB_{0}}\breve{\bB_{0}}^{\tp})$, where $\trace(\bA)$ denotes the trace of a square matrix $\bA$; 
(ii) $\left[\bI_{Q} \otimes \bI_{P} \otimes \left( \theta \int \bB''(t)\bB''(t)^{\tp}dt + \phi\bI_{H} \right)^{1/2}\right] \vect(\sfB_{0})$ $= \vect\left( (\bI_{P} \otimes \left( \theta \int \bB''(t)\bB''(t)^{\tp}dt + \phi\bI_{H} \right)^{1/2})\breve{\bB_{0}} \bI_{Q}\right)$.

Therefore, equivalently, the optimization problem reduces to an unregulated least-squares problem with modified predictor and outcome variables. To get an estimate of $\sfB_{0}$ use
$\hat{\sfB}_{0} = \arg\min_{\rank(\sfB_{0}) \leq R} \frac{1}{NJ}\int_{\sT}\| \tilde{\sY} - \langle \tilde{\sZ} , {\sfB_{0}} \rangle_{L+1} \|^{2}dt$.
where $\tilde{\sZ} \in \bbR^{(NJ + HP) \times H \times P_{1} \times \cdots \times P_{L}\times Q_{1}\times\cdots Q_{M}}$ and 
$\tilde{\sY} \in \bbR^{(NJ + HP) \times Q_{1} \times \cdots \times Q_{M}}$ such that
the unfolding of $\tilde{\sZ}$ and $\tilde{\sY}$ along the first dimension produces the following matrices:
\begin{equation}
\label{eq:short-tensor}
\tilde{\sZ}_{(1)} = 
\begin{bmatrix}
\sZ_{(1)} \\
(\bI_{P} \otimes \left( \theta \int \bB''(t)\bB''(t)^{\tp}dt + \phi\bI_{H} \right)^{1/2})
\end{bmatrix}    
\end{equation}
and 
\begin{equation}
\tilde{\sY}_{(1)} = 
\begin{bmatrix}
\sY_{(1)} \\
0_{HP \times Q}\\
\end{bmatrix}
\end{equation}
Therefore, $\tilde{\sZ}$ be the contamination of $\sfZ(t)$ along  with smoothing term and the sparsity, and $\tilde{\sY}$ is a contamination of $\sfY(t)$ and the zero tensor function.
Thus, apply the following Algorithm \ref{algo:estimate-beta} to get the estimate of coefficient tensor for known rank of the coefficient array and hence the coefficient function $\bbeta(t)$.

\begin{algorithm}
\caption{Estimation of $\bbeta(t)$}
\label{algo:estimate-beta}
\begin{algorithmic}
\STATE \textbf{Data:} $\sfX(t)$, $\sfY(t)$ for $t \in \left[0, T\right], T > 0$ observed on a grid in $[0,T]$.
\STATE \textbf{Tuning parameters:} $\lbrace\theta, \phi\rbrace$, 
rank $R\in\mathbb{N}$, number of knots $K_{N}$, a vector of known B-spline bases
$\sB(t) = (\bbB_{1}(t), \cdots, \bbB_{H}(t))^{\tp}$.
\STATE \textbf{Stopping parameter:} $\epsilon_{0} > 0$
\STATE \textbf{Create:} $\sZ$ and $\sY$ as mentioned in Equation (\ref{eq:short-tensor}).
\STATE \textbf{Initialize:} $\bU_{0}, \bU_{1}, \cdots, \bU_{L}, \bV_{1}, \cdots, \bV_{M}$ be randomly chosen matrices of specific order.
\WHILE{$\text{Error} > \epsilon_{0}$}
\FOR{$l \gets 1$ to $\#\lbrace H, P_{1}, \cdots, P_{L}\rbrace$}
    \STATE \textbf{Set} $d^{(l)}$ be the $l$-th entry of $\{H, P_{1}, \cdots, P_{L}\}$
    \FOR{$r = 1, \cdots, R$}
        \STATE $\sfC_{r} \gets \langle \tilde{\sZ}, \bu_{r0}\circ \cdots\circ \bu_{r,k-1} \circ \bu_{r,k+1} \circ \cdots \circ \bu_{rL} \circ \bv_{r1} \circ \cdots \circ \bv_{rM}\rangle_{L}$ which is a tensor of dimension $(NJ+ HP) \times d^{(l)} \times Q_{1} \times \cdots \times Q_{M}$
        \STATE Unfolding $\sfC_{r}$ along with dimension corresponding to $d^{(l)}$ 
        \STATE Obtain a $(NJ + HP)Q \times d^{(l)}$ dimension matrix $\bC_{r}$
    \ENDFOR
    \STATE $\bC \gets \left[\bC_{1}, \cdots, \bC_{R}\right] \in \bbR^{(NJ+HP)Q \times Rd^{(l)}}$
    \STATE $\vect(\bU_{l}) \gets (\bC^{\tp}\bC)^{-1}\bC^{\tp}\vect(\tilde{\sY})$ 
\ENDFOR
\FOR{$m \gets 1$ to $\#\lbrace Q_{1}, \cdots, Q_{M}\rbrace$}
    \STATE \textbf{Set} $d^{(m)}$ be the $m$-th entry of $\{Q_{1}, \cdots, Q_{L}\}$
    \STATE $\tilde{\sY}_{d^{(m)}}$ is unfolded along the mode corresponding to $d^{(m)}$ and obtain a $d^{(m)} \times (NJ+HP)\prod_{m \neq k}Q_{m}$
    \FOR{$r = 1, \cdots, R$}
        \STATE $D_{r} \gets \vect(\langle \tilde{\sZ}, \bu_{r0}\circ \bu_{r1} \circ \cdots \circ \bu_{rL} \circ \bv_{r1} \circ \cdots \bv_{r,k-1} \circ \bv_{r,k+1} \circ \cdots \circ \bv_{rM}\rangle_{L+1})$
    \ENDFOR
    \STATE $\bD \gets \left[D_{1}, \cdots, D_{R}\right] \in \bbR^{(NJ+HP)\prod_{m \neq k}Q_{m} \times R}$
    \STATE $\bV_{m} \gets \tilde{\sY}_{d^{(m)}}\bD(\bD^{\tp}\bD)^{-1}$
\ENDFOR
\STATE \textbf{Compute}  $\bB = [[ \bU_{0}, \bU_{1}, \cdots, \bU_{L}, \bV_{1}, \cdots, \bV_{M}]]$ 
\STATE \textbf{Calculate} $\text{Error} = \frac{\|\tilde{\sY} - \left<\tilde{\sZ}, \hat{\bB} \right>_{L+1} \|_{\sF}^{2}}{\|\hat{\sY}\|_{\sF}^{2}}$ 
\ENDWHILE
\STATE \textbf{Compute} $\beta_{p_{1}, \cdots, p_{L} , q_{1}, \cdots, q_{M}}(t) = \bb_{p_{1}, \cdots, p_{L} , q_{1}, \cdots, q_{M}}^{\tp}\sB(t)$
using Equation (\ref{eq:basisExp}) for each node
\end{algorithmic}
\end{algorithm}


\section{Simulation}
\label{sec:simulation}
In this section, we conduct numerical studies to compare the finite sample performance to estimate the four-way time-varying tensor coefficient $\bbeta(t)$. Data are generated from the following model, for each mode $p_{1}, p_{2}, q_{1}, q_{2}$
\begin{align*}
    &y_{i, q_{1}, q_{2}}(t) = \sum_{p_{1} = 1}^{P_{1}}\sum_{p_{2} = 1}^{P_{2}}x_{i, p_{1}, p_{2}}(t)\beta_{p_{1}, p_{2}, q_{1}, q_{2}}(t) + \epsilon_{i, q_{1}, q_{2}}(t),\\
    &i = 1, \cdots, N; t \in [0, 1]
    \numberthis
\end{align*}
The regression functions are given by 
\begin{align*}
 \beta_{ p_{1}, p_{2}, q_{1}, q_{2}}(t) &= 
 p_{1}\cos\left({2\pi t}\right) + 
 q_{1}\sin\left({2\pi t}\right) \\
  & \qquad +
 p_{2}\sin\left({4\pi t}\right) + 
 q_{2}\cos\left({4\pi t}\right)   
\end{align*}
Here, changes in one unit of the index of each mode produce a change in one unit of the coefficient when the time is fixed. The covariates are generated as follows:
    $x_{i, p_{1}, p_{2}}(t) = \chi_{i, p_{1}, p_{2}}^{(1)} + \chi_{i, p_{1}, p_{2}}^{(2)}\sin\left({\pi t}\right) + 
         \chi_{i, p_{1}, p_{2}}^{(3)}\cos\left({\pi t}\right)$
and the errors are generated as follows:
    $\epsilon_{i, q_{1}, q_{2}}(t) = \eta_{i, q_{1}, q_{2}}^{(1)}\sqrt{2}\cos\left({\pi t}\right) 
    + \eta_{i, q_{1}, q_{2}}^{(2)}\sqrt{2}\sin\left({\pi t}\right)$
for all $p_{1} = 1, \cdots, P_{1}, p_{2} = 1, \cdots, P_{2}, q_{1} = 1,\cdots, Q_{1}$ and $q_{2} = 1, \cdots, Q_{2}$. 
Moreover, we assume that $x_{i, p_{1}, p_{2}}(t)$ are observed with measurement error, i.e., $u_{i, p_{1}, p_{2}}(t) = x_{i, p_{1}, p_{2}} + \delta_{p_{1}, p_{2}}$ where $\delta_{p_{1}, p_{2}} \sim N(0, 0.6^2)$. Assume that the set of random variables $\lbrace \chi_{i, p_{1}, p_{2}}^{(l)}: l = 1, 2, 3\rbrace$ and $\lbrace \eta_{i, q_{1}, q_{2}}^{(l)}: l = 1, 2\rbrace$ is mutually independent. 
The data generating process is influenced by \cite{kim2018additive} which has been used in different concepts.
We observe the data at 81 equidistant time points in $[0, 1]$ with $t_{j} = (j - 0.5)/J$ for all $j = 1, \cdots, J$. 
We also fix $P_{1} \times P_{2} = 5\times 2$ and $Q_{1}\times Q_{2}$ be either $5\times 2$ or $15\times 12$.
Set, number of subjects, $N \in \{ 30, 100\}$. 
We consider the following scenarios:
\begin{itemize}[align = left]
    \item Situation-1: We choose $\chi_{i, p_{1}, p_{2}}^{(1)} \sim N(0, 1^{2})$, $\chi_{i, p_{1}, p_{2}}^{(2)} \sim N(0, 0.85^{2})$, $\chi_{i, p_{1}, p_{2}}^{(3)} \sim N(0, 0.7^{2})$ and they are mutually independent. $\eta_{i, q_{1}, q_{2}}^{(1)} \sim N(0, 2^{2})$, $\eta_{i, q_{1}, q_{2}}^{(2)} \sim N(0, 0.75^{2})$ and they are mutually independent. Here, the covariates do not depend on the modes of the data structure. 
    \item Situation-2: In addition with the assumption of the coefficients of covariates, impose the spatial correlation structure to address the mode-wise dependencies. We consider the following two cases. 
    (a) 
$\chi^{(l)}_{i, p_{1}, p_{2}}$ at mode $(p_{1}, p_{2})$ is 
$\rho_{s}(\text{ED}_{p_{1}, p_{2}}; \theta)$, 
where $\rho_{s}$ is the exponential correlation function, $\text{ED}_{p_{1}, p_{2}}$ is defined as scaled Euclidean distance between two modes, having been scaled by a constant $\theta$, 
therefore, $\theta$ defines an isotropic covariance function. In this simulation setup, $\theta$ is taken as 8. 
(b) 
$\chi^{(l)}_{i, p_{1}, p_{2}}$ at mode $(p_{1}, p_{2})$ is 
$\rho_{M}(d_{p_{1}, p_{2}}; \kappa, \nu)$, where $d_{q_{1}, q_{2}}$ denotes the Euclidean distance between two different modes and $\rho_{M}$ is the correlation function, belongs to  Mat\'ern family. The Mat\'ern isotropic auto-correlation function has a specific form 
$\rho_{M}(d; \kappa, \nu) = \frac{2^{1-\nu}}{\Gamma(\nu)}\left( \frac{2d\sqrt{\nu}}{\kappa}\right)^{\nu}K_{\nu}\left(\frac{2d\sqrt{\nu}}{\kappa}\right)$,
for $\kappa, \nu > 0$. Here, $K_{\nu}(\cdot)$ is termed as Bessel function of order $\nu$. 
The positive range parameter $\kappa$ controls the decay of the correlation between the observations at a large distance $d$. 
The order $\nu$ controls the behavior of autocorrelation function for the observations which are separated by small distance. 
For our numerical example, we set scale $\kappa = 0.55$ and the smoothness parameter $\nu = 1$. This was implemented using ``\texttt{stationary.image.cov}'' and ``\texttt{matern.image.cov}''
functions respectively available in \texttt{fields} package in R \cite{fields}. 
\end{itemize}
\par
We ran the simulation 100 times for each scenario to evaluate our method. 
For each of the simulation setups, we set the number of knots as $[J/4]$, where $[a]$ denotes the integer part of $a$. We compare the overall performance of the models to estimate the parameter curves for different choices of ranks by studying several error rates based on different norms. 
We choose smoothing parameters $\theta$ from the set $\{0, 0.001, 0.005, 0.01, 0.05, 0.1\}$, 
and $\phi$ from the set $\{0, 0.5, 3, 10\}$, and allow values from 1 to 5 for the choice of rank $R$. In the following tables, we denote the proposed functional tensor-on-tensor model with rank $r$ as $\text{FToTM}_{r}$.
To compare with the existing literature, we apply the concurrent linear model \cite{ramsay2005springer} (CLM) for mode-wise analysis and implement this method using the ``\texttt{pffr}" function available in the \texttt{refund} \cite{refund} package in R, with the penalized concurrent effect of functional covariates \cite{ivanescu2015penalized}. 
\par
Tables \ref{table:IMSE-independent}, \ref{table:IMSE-X-dep-stationary} and \ref{table:X-dep-Matern}
show the results of integrated and relative integrated  mean square errors which are defined as 
$\text{IMSE} = \int_{t \in \mathcal{T}} \|\hat{\bbeta}(t) - \bbeta(t) \|_{\sF}^{2}dt$ 
and $\text{RIMSE} = \frac{\int_{t \in \mathcal{T}} \|\hat{\bbeta}(t) - \bbeta(t) \|_{\sF}^{2}dt}{\int_{t \in \mathcal{T}} \|\bbeta(t) \|_{\sF}^{2}dt}$, respectively. 
Similarly, we report the absolute integrated and relative integrated mean square errors which are 
$\text{IMAE} = \int_{t \in \mathcal{T}} \sum_{{p_{1}, p_{2}, q_{1}, q_{2}}}\left|\hat{\beta}_{p_{1}, p_{2}, q_{1}, q_{2}}(t) - \beta_{p_{1}, p_{2}, q_{1}, q_{2}}(t) \right|dt$ 
and \newline $\text{RIMAE} = \frac{\int_{t \in \mathcal{T}} \sum_{{p_{1}, p_{2}, q_{1}, q_{2}}}\left|\hat{\beta}_{p_{1}, p_{2}, q_{1}, q_{2}}(t) - \beta_{p_{1}, p_{2}, q_{1}, q_{2}}(t) \right|dt}{\int_{t \in \mathcal{T}} \sum_{{p_{1}, p_{2}, q_{1}, q_{2}}} \left|\beta_{p_{1}, p_{2}, q_{1}, q_{2}}(t) \right|dt}$, respectively. 
The advantage of these simulation situations are that these models are not based on the reduced-rank model. 
Here, we observe the curves in the presence of errors.
All integrals are approximated using the Riemann sum. 
Since our proposed method involves an iterative procedure which depends on the initial estimates, the computational time is therefore not comparable to that of the classical CLM, which is not an iterative method.
For all situations, our proposed method does a much better job in terms of low error rates in estimating the parameter $\bbeta(t)$. 

\begin{table*}[!t]
\centering
\caption{Results of simulation situation-1 where each modes are assumed to be independent for $\sfX(t)$ and $\sfE(t)$ for fixed time-points. 
    Here we assume each of $\{\chi_{p_{1}, p_{2}}^{(k)}\}_{p_{1}, p_{2}}$
    and 
    $\{\eta_{q_{1}, q_{2}}\}_{q_{1}, q_{2}}^{(k)}$ are independent 
    for $(p_{1}, p_{2})$ and $(q_{1}, q_{2})$ respectively.\label{table:IMSE-independent}}
\begin{tabular}{@{}ccccccc@{}}
\hline
Method & IMSE (SD) & RIMSE (SD) & IMAE (SD) & RIMAE (SD)\\
\hline
    \multicolumn{5}{c}{$N = 30, P_{1} \times P_{2} = 5 \times 2, Q_{1} \times Q_{2} = 5 \times 2$}\\
    \hline
    CLM & 0.14294 (0.02046) 
            & 0.01059 (0.00152) 
            & 0.28311 (0.02027) 
            & 0.09244 (0.00662) \\ 
    $\text{FToTM}_{1}$ & 1.48469 (0.05628) 
                    & 0.10998 (0.00417)
                    & 0.96636 (0.01626) 
                    & 0.31552  (0.00531) \\ 
    $\text{FToTM}_{2}$ & 0.45773 (0.02218) 
                    & 0.03391 (0.00164) 
                    & 0.53786 (0.01068) 
                    & 0.17561 (0.00349) \\ 
    $\text{FToTM}_{3}$ & 0.15078 (0.01316) 
                    & 0.01117 (0.00097) 
                    & 0.29482 (0.01452) 
                    & 0.09626 (0.00474) \\ 
    $\text{FToTM}_{4}$ & 0.01065 (0.00383)
                    & 0.00079 (0.00028) 
                    & 0.07871 (0.01367) 
                    & 0.0257 (0.00446) \\ 
    $\text{FToTM}_{5}$ & 0.01558 (0.00582) 
                    & 0.00115 (0.00043) 
                    & 0.09412 (0.01695) 
                    & 0.03073 (0.00553) \\ 
    \hline
    \multicolumn{5}{c}{$N = 30, P_{1} \times P_{2} = 5 \times 2, Q_{1} \times Q_{2} = 15 \times 12$}\\     
    \hline
    CLM & 0.1448  (0.01339) 
        & 0.00193 (0.00018) 
        & 0.28468 (0.0132) 
        & 0.04054 (0.00188) \\ 
    $\text{FToTM}_{1}$ & 9.24824 (0.06732) 
                     & 0.12304 (9e-04) 
                     & 2.27313 (0.01304) 
                     & 0.32372 (0.00186) \\ 
    $\text{FToTM}_{2}$ & 1.79804 (0.06786) 
                     & 0.02392 (9e-04) 
                     & 1.02121 (0.01836) 
                     & 0.14543 (0.00261) \\ 
    $\text{FToTM}_{3}$ & 0.23289 (0.02089) 
                    & 0.0031 (0.00028) 
                    & 0.36104 (0.01293) 
                    & 0.05142 (0.00184) \\ 
    $\text{FToTM}_{4}$ & 0.06108 (0.06808) 
                    & 0.00081 (0.00091) 
                    & 0.15243 (0.13744) 
                    & 0.02171 (0.01957) \\ 
    $\text{FToTM}_{5}$ & 0.00195 (0.00053) 
                    & 3e-05 (1e-05) 
                    & 0.03348 (0.00451) 
                    & 0.00477 (0.00064) \\ 
    \hline
    \multicolumn{5}{c}{$N = 100, P_{1} \times P_{2} = 5 \times 2, Q_{1} \times Q_{2} = 5 \times 2$}\\  
    \hline
    CLM & 0.03087 (0.00348) 
        & 0.00229 (0.00026) 
        & 0.13236 (0.00731) 
        & 0.04322 (0.00239) \\ 
    $\text{FToTM}_{1}$ & 1.46268 (0.04068) 
                    & 0.10835 (0.00301) 
                    & 0.95921 (0.01095) 
                    & 0.31319 (0.00358) \\ 
    $\text{FToTM}_{2}$ & 0.43737 (0.01418) 
                    & 0.0324 (0.00105) 
                    & 0.52551 (0.00725) 
                    & 0.17158 (0.00237) \\ 
    $\text{FToTM}_{3}$ & 0.13651 (0.00541) 
                    & 0.01011 (4e-04) 
                    & 0.27253 (0.01099) 
                    & 0.08898 (0.00359) \\ 
    $\text{FToTM}_{4}$ & 0.00303 (0.00091) 
                    & 0.00022 (7e-05)
                    & 0.04222 (0.00632) 
                    & 0.01379 (0.00206) \\ 
    $\text{FToTM}_{5}$ & 0.0037 (0.00115) 
                    & 0.00027 (8e-05) 
                    & 0.04663 (0.00696) 
                    & 0.01523 (0.00227) \\ 
    \hline
    \multicolumn{5}{c}{$N = 100, P_{1} \times P_{2} = 5 \times 2, Q_{1} \times Q_{2} = 15 \times 12$}\\  
    \hline
    CLM & 0.03082 (0.00163) 
        & 0.00041 (2e-05) 
        & 0.1328 (0.00357) 
        & 0.01891 (0.00051) \\ 
    $\text{FToTM}_{1}$ & 9.21298 (0.04487) 
                    & 0.12257 (6e-04) 
                    & 2.26689 (0.01132) 
                    & 0.32283 (0.00161) \\ 
    $\text{FToTM}_{2}$ & 1.76018 (0.04482) 
                    & 0.02342 (6e-04) 
                    & 1.00917 (0.01218) 
                    & 0.14372 (0.00173) \\ 
    $\text{FToTM}_{3}$ & 0.22276 (0.03467) 
                    & 0.00296 (0.00046) 
                    & 0.35168 (0.02647) 
                    & 0.05008 (0.00377) \\ 
    $\text{FToTM}_{4}$ & 0.05837 (0.06468) 
                    & 0.00078 (0.00086) 
                    & 0.14918 (0.14726) 
                    & 0.02124 (0.02097) \\ 
    $\text{FToTM}_{5}$ & 0.00085 (0.00033) 
                    & 1e-05 (0) 
                    & 0.02197 (0.00403) 
                    & 0.00313 (0.00057) \\                 
\hline
\end{tabular}
\end{table*}

\begin{table*}[!t]
\centering
\caption{
Results of simulation situation-2(a) 
where each modes are assumed to be independent for $\sfE(t)$ for fixed time-points whereas modes for $\sfX(t)$ are assumed to be dependent. 
Here we assume $\{\chi_{p_{1}, p_{2}}^{(k)}\}_{p_{1}, p_{2}}$ is spatially dependent with exponential covariance function. 
\label{table:IMSE-X-dep-stationary}}
\begin{tabular}{@{}ccccccc@{}}
\hline
Method & IMSE (SD) & RIMSE (SD) & IMAE (SD) & RIMAE (SD)\\
\hline
    \multicolumn{5}{c}{$N = 30, P_{1} \times P_{2} = 5 \times 2, Q_{1} \times Q_{2} = 5 \times 2$}\\
    \hline
    CLM & 11.03513 (2.27364) 
        & 0.81742 (0.16842) 
        & 2.45079 (0.24082) 
        & 0.8002 (0.07863) \\ 
    $\text{FToTM}_{1}$ & 1.46631  (0.0141) 
                        & 0.10862 (0.00104) 
                        & 0.96402 (0.00583) 
                        & 0.31476 (0.0019) \\ 
    $\text{FToTM}_{2}$ & 0.60273 (0.01917) 
                        & 0.04465 (0.00142) 
                        & 0.60152 (0.01318) 
                        & 0.1964 (0.0043) \\ 
    $\text{FToTM}_{3}$ & 0.32753 (0.01741) 
                        & 0.02426 (0.00129) 
                        & 0.42707 (0.01962) 
                        & 0.13944 (0.00641) \\ 
    $\text{FToTM}_{4}$ & 0.21328 (0.21078) 
                        & 0.0158 (0.01561) 
                        & 0.35394 (0.13306) 
                        & 0.11556 (0.04344) \\ 
    $\text{FToTM}_{5}$ & 0.13694 (0.02654) 
                        & 0.01014 (0.00197) 
                        & 0.30854 (0.0384) 
                        & 0.10074 (0.01254) \\ 
    \hline      
    \multicolumn{5}{c}{$N = 30, P_{1} \times P_{2} = 5 \times 2, Q_{1} \times Q_{2} = 15 \times 12$}\\
    \hline
    CLM & 11.36335 (1.34533) 
        & 0.15118 (0.0179) 
        & 2.49712 (0.14845) 
        & 0.35562 (0.02114) \\ 
    $\text{FToTM}_{1}$ & 9.21977 (0.02778) 
                        & 0.12266 (0.00037) 
                        & 2.27091 (0.0079) 
                        & 0.32341 (0.00112) \\ 
    $\text{FToTM}_{2}$ & 1.76995 (0.02734) 
                        & 0.02355 (0.00036) 
                        & 1.01769 (0.01081) 
                        & 0.14493 (0.00154) \\ 
    $\text{FToTM}_{3}$ & 0.41264 (0.16057) 
                        & 0.00549 (0.00214) 
                        & 0.48365 (0.08798) 
                        & 0.06888 (0.01253) \\ 
    $\text{FToTM}_{4}$ & 0.18218 (0.21293) 
                        & 0.00242 (0.00283) 
                        & 0.32063 (0.13906) 
                        & 0.04566 (0.0198) \\ 
    $\text{FToTM}_{5}$ & 0.06936 (0.05182) 
                        & 0.00092 (0.00069) 
                        & 0.19811 (0.08864) 
                        & 0.02821 (0.01262) \\              
    \hline
    \multicolumn{5}{c}{$N = 100, P_{1} \times P_{2} = 5 \times 2, Q_{1} \times Q_{2} = 5 \times 2$}\\
    \hline                
    CLM & 2.55232 (0.45649) 
        & 0.18906 (0.03381) 
        & 1.19172 (0.10708) 
        & 0.38911 (0.03496) \\ 
    $\text{FToTM}_{1}$ & 1.45974 (0.00711) 
                      & 0.10813 (0.00053) 
                      & 0.96178 (0.00323) 
                      & 0.31403 (0.00105) \\ 
    $\text{FToTM}_{2}$ & 0.58776 (0.01049) 
                      & 0.04354 (0.00078) 
                      & 0.59246 (0.00766) 
                      & 0.19344 (0.0025) \\ 
    $\text{FToTM}_{3}$  & 0.31275 (0.00961) 
                    & 0.02317 (0.00071) 
                    & 0.411 (0.01063) 
                    & 0.1342 (0.00347) \\ 
    $\text{FToTM}_{4}$ & 0.18492 (0.20235) 
                    & 0.0137 (0.01499) 
                    & 0.32604 (0.13409)
                    & 0.10646 (0.04378) \\ 
    $\text{FToTM}_{5}$ & 0.11149 (0.03648) 
                    & 0.00826 (0.0027) 
                    & 0.27665 (0.06128) 
                    & 0.09033 (0.02001) \\          
    \hline
    \multicolumn{5}{c}{$N = 100, P_{1} \times P_{2} = 5 \times 2, Q_{1} \times Q_{2} = 15 \times 12$}\\
    \hline                         
    CLM & 2.5259 (0.21061) 
        & 0.0336 (0.0028) 
        & 1.18808 (0.05122) 
        & 0.1692 (0.00729) \\ 
    $\text{FToTM}_{1}$ & 9.26995 (0.13929) 
                    & 0.12333 (0.00185) 
                    & 2.28525 (0.03385) 
                    & 0.32545 (0.00482) \\ 
    $\text{FToTM}_{2}$ & 1.74798 (0.01575) 
                    & 0.02325 (0.00021) 
                    & 1.00948 (0.00691) 
                    & 0.14376 (0.00098) \\ 
    $\text{FToTM}_{3}$ & 0.61359 (0.30173) 
                    & 0.00816 (0.00401) 
                    & 0.58812 (0.16308) 
                    & 0.08376 (0.02322) \\ 
    $\text{FToTM}_{4}$ & 0.66733 (0.41716) 
                    & 0.00888 (0.00555) 
                    & 0.596 (0.24026) 
                    & 0.08488 (0.03422) \\ 
    $\text{FToTM}_{5}$ & 0.0914 (0.04684) 
                    & 0.00122 (0.00062) 
                    & 0.23906 (0.07987) 
                    & 0.03405 (0.01137) \\                 
\hline
\end{tabular}
\end{table*}

\begin{table*}[!t]
\centering
\caption{Results of simulation situation-2(b) 
where each modes are assumed to be independent for $\sfE(t)$ for fixed time-points whereas modes for $\sfX(t)$ are assumed to be dependent. 
Here we assume $\{\chi_{p_{1}, p_{2}}^{(k)}\}_{p_{1}, p_{2}}$ is spatially dependent with Mat\'ern covariance function. \label{table:X-dep-Matern}}
\begin{tabular}{@{}ccccccc@{}}
\hline
Method & IMSE (SD) & RIMSE (SD) & IMAE (SD) & RIMAE (SD)\\
\hline
\multicolumn{5}{c}{$N = 30, P_{1} \times P_{2} = 5 \times 2, Q_{1} \times Q_{2} = 5 \times 2$}\\
\hline
    CLM & 0.26393 (0.04919) 
        & 0.01955 (0.00364) 
        & 0.38374 (0.03318) 
        & 0.12529 (0.01083) \\ 
  $\text{FToTM}_{1}$ & 1.45885 (0.02061) 
                    & 0.10806  (0.00153) 
                    & 0.9599  (0.00731) 
                    & 0.31342 (0.00239)\\ 
  $\text{FToTM}_{2}$ & 0.46879 (0.02445) 
                    & 0.03473 (0.00181) 
                    & 0.54097 (0.01118) 
                    & 0.17663 (0.00365) \\ 
  $\text{FToTM}_{3}$ & 0.16291 (0.01629) 
                    & 0.01207 (0.00121) 
                    & 0.30998 (0.01506) 
                    & 0.10121 (0.00492) \\ 
  $\text{FToTM}_{4}$ & 0.0087 (0.01146) 
                    & 0.00064 (0.00085) 
                    & 0.06782 (0.0274) 
                    & 0.02214 (0.00895) \\ 
  $\text{FToTM}_{5}$ & 0.0111 (0.00525) 
                    & 0.00082 (0.00039) 
                    & 0.07909 (0.01855) 
                    & 0.02582 (0.00606) \\ 
    \hline
    \multicolumn{5}{c}{$N = 30, P_{1} \times P_{2} = 5 \times 2, Q_{1} \times Q_{2} = 15\times 12$}\\
    \hline
    CLM & 0.26313 (0.02958)
        & 0.0035 (0.00039) 
        & 0.3835 (0.02167) 
        & 0.05462 (0.00309) \\ 
  $\text{FToTM}_{1}$ & 9.22145  (0.02791) 
                    & 0.12268 (0.00037) 
                    & 2.27063 (0.00894) 
                    & 0.32337 (0.00127) \\ 
  $\text{FToTM}_{2}$ & 1.77848 (0.02878) 
                    & 0.02366 (0.00038) 
                    & 1.02026 (0.01052) 
                    & 0.1453 (0.0015) \\ 
  $\text{FToTM}_{3}$ & 0.23293 (0.01206) 
                    & 0.0031 (0.00016) 
                    & 0.36047 (0.00952) 
                    & 0.05134 (0.00136) \\ 
  $\text{FToTM}_{4}$ & 0.05929 (0.06315) 
                    & 0.00079 (0.00084) 
                    & 0.15872 (0.13817) 
                    & 0.0226 (0.01968) \\ 
  $\text{FToTM}_{5}$ & 0.00175 (0.00133) 
                    & 2e-05 (2e-05) 
                    & 0.03081 (0.00931) 
                    & 0.00439 (0.00133) \\ 
    \hline
    \multicolumn{5}{c}{$N = 100, P_{1} \times P_{2} = 5 \times 2, Q_{1} \times Q_{2} = 5\times 2$}\\
    \hline       
    CLM & 0.05833 (0.00912) 
        & 0.00432 (0.00068) 
        & 0.18217 (0.01463) 
        & 0.05948 (0.00478) \\ 
  $\text{FToTM}_{1}$ & 1.44275 (0.00963) 
                    & 0.10687 (0.00071) 
                    & 0.95499 (0.00374) 
                    & 0.31181 (0.00122) \\ 
  $\text{FToTM}_{2}$ & 0.44676 (0.01346) 
                    & 0.03309 (0.001) 
                    & 0.52798 (0.00559) 
                    & 0.17239 (0.00183) \\ 
  $\text{FToTM}_{3}$ & 0.14999 (0.00779) 
                    & 0.01111 (0.00058) 
                    & 0.29657 (0.00869) 
                    & 0.09683 (0.00284) \\ 
  $\text{FToTM}_{4}$ & 0.00231 (0.00143) 
                    & 0.00017 (0.00011) 
                    & 0.03593 (0.00956) 
                    & 0.01173 (0.00312) \\ 
  $\text{FToTM}_{5}$ & 0.00284 (0.00125) 
                    & 0.00021 (9e-05) 
                    & 0.04026 (0.00816) 
                    & 0.01314 (0.00266) \\ 
    \hline    
    \multicolumn{5}{c}{$N = 100, P_{1} \times P_{2} = 5 \times 2, Q_{1} \times Q_{2} = 15\times 12$}\\
    \hline
    CLM & 0.05746 (0.00427) 
        & 0.00076 (6e-05) 
        & 0.18093 (0.00695) 
        & 0.02577 (0.00099) \\ 
  $\text{FToTM}_{1}$ & 9.18754 (0.00744) 
                    & 0.12223 (1e-04) 
                    & 2.26337 (0.00385) 
                    & 0.32233 (0.00055) \\ 
  $\text{FToTM}_{2}$ & 1.73663 (0.00773) 
                    & 0.0231 (1e-04) 
                    & 1.00481 (0.0038) 
                    & 0.1431 (0.00054) \\ 
  $\text{FToTM}_{3}$ & 0.2181 (0.00522) 
                    & 0.0029  (7e-05) 
                    & 0.34535 (0.00306) 
                    & 0.04918 (0.00044) \\ 
  $\text{FToTM}_{4}$ & 0.05167 (0.05987) 
                    & 0.00069 (8e-04) 
                    & 0.13999 (0.14339) 
                    & 0.01994 (0.02042) \\ 
  $\text{FToTM}_{5}$ & 0.00081 (0.00061) 
                    & 1e-05 (1e-05)
                    & 0.02055 (0.00654) 
                    & 0.00293 (0.00093) \\ 
\hline
\end{tabular}
\end{table*}

\section{Application to \textit{ForrestGump} data set}
\label{sec:real-data}
The Studyforrest (website: \url{https://www.studyforrest.org/}) describes a publicly available dataset for the study of neural language and story processing. The imaging data analyzed in this paper is publicly available through OpenfMRI (\url{https://openneuro.org/datasets/ds000113/versions/1.3.0}) \cite{hanke2014high, sengupta2016studyforrest}. 
In total 15 right-handed participants (mean age 29.4 years, range 21–39, 40\% females, native German speaker) 
volunteered for a series of studies including  eye-tracking experiments using natural signal stimulation with a motion picture.
Volunteers have no known hearing problem without permanent or current temporary impairments and no neurological disorder.
Participants viewed a feature film ``Forrest Gump" (Robert Zemeckis, Paramount Pictures, 1994 with German audio track) in eight back-to-back 15 minute long movie sessions. 
The eye tracking camera was fitted just outside the scanner bore, approximately centered, and viewing the left eye of the participant at a distance of 100 cm through a small gap between the top of the back projection screen and the scanner bore ceiling.
Participants were allowed to perform free eye movements without requiring to fixate or keep the eye open. 
The eye gaze recording started as soon as the computer received the first fMRI trigger signal. 
\par
The normalized eye-gaze coordinate time series contain the X and Y coordinates of the eye-gaze, pupil area measurements, and the corresponding numerical ID of the movie frame presented at the time of measurement are obtained. 
In the eye-gazing data, there is significant loss of information due to eye blinks, and those are marked as NaN in the data set and imputed via spline interpolation. 
We use 14 individuals and remove Subject 5 due to excessive missing data.
To analyze the data on a local computer, we only used the first run of the experiment for each individual and down-sampled the images to $64\times 64 \times 64$ via nearest-neighbor interpolation where the number of time-points was 451 (first one-eighth of the movie). 
Details of the pre-processing steps performed along with further information of data acquisitions are described in Appendix.
\par
Our scientific question of interest was to understand the association between brain image pattern in the presence of audio-visual inputs. 
This is the first approach to statistically analyze such a study by exploiting the complex structure of the data. 
We use the eye position in an angular unit (i.e., polar coordinates) instead of the Cartesian coordinates, where we report magnitude changes of eye position in the screen reference system. 
Furthermore, the X and Y coordinates, the related polar coordinates, and the pupil area were down-sampled to match the fMRI sampling frequency. 
We fit a time-varying tensor regression coefficient model as described in Section \ref{sec:general}. 
Our covariate is a 3-mode tensor representing normalized eye-gaze coordinate time-series; each  mode represents scaled polar coordinates of the eye-gaze and pupil area measurements, respectively. The response of the model is pre-processed fMRI data. Response and covariates are collected simultaneously.
The coefficient functions $\bbeta_{1}$, $\bbeta_{2}$ and $\bbeta_{3}$ are amplitudes over the time associated with distance, angle of eye-gaze and pupil area, respectively; included to detect the effect of movie in a visual form in BOLD response change. 
We choose the rank for reduced-rank extraction to be 3 since it has the lowest prediction error.
\par
For interpretation purposes, we evaluate estimates
$\hat{\bbeta}(t)$ by taking average values over eight different functional networks in the brain. This was achieved by first parcellating the brain into the 268 regions of the Shen atlas \cite{shen2013groupwise}. These regions were thereafter further combined into eight functional networks 
\cite{finn2015functional}:
medial frontal, frontoparietal, default mode, subcortical-cerebellum, motor, visual I, visual II, and visual association.
Figure \ref{fig:coef} 
represents the average estimated coefficient function corresponding to three
visual features (distance, angle of eye-gaze, and pupil area) over all the time-points for each network, respectively. Throughout the time course changes in visual features has greatest impact on activation in visual I, depicted using purple lines, which should be expected as participants view the movie. 
Vertical lines represent scene changes in the movie. The first segment, consisting of approximately 84 time-points corresponds to the opening sequence, which shows a feather floating through the sky as credits are shown. The second segment consists of the famous scene where the protagonist of the movie sits on a bench at a bus stop and begins discussing the story of his life. During this scene, there is heightened activation in several brain networks in reaction to  different visual features. Subsequent segments represent scene changes alternating between interior and exterior settings; see 
\cite{hausler2016annotation}
for more details.

\section{Discussion}
\label{sec:discussion}
In this paper, we have proposed a time-varying tensor-on-tensor regression model and a method to estimate the coefficient tensors which belong to an infinite-dimensional space. We believe the method provides an efficient approach towards performing multi-modal data analysis using neuroimaging data.
Regression coefficients are expressed using the B-spline technique, and the coefficients of the B-spline bases are estimated using low-rank  tensor decomposition. 
This method reduces the vastness of the parameters of interest and computational complexity. 
We have provided a meaningful simulation study, as well as performed real data analysis combining fMRI and eye-tracking data. The results of our data analysis suggests the approach has promise for identifying brain regions responding to an external stimulus, which in this case is movie watching.\par
Although our tensor data can be compactly represented by a CP model, it is NP hard to determine the rank of the low-rank decomposition
\cite{johan1990tensor}. 
To determine the tuning parameters, one can perform the cross-validation technique. 
However, our main objective is not to choose the optimal rank of the low-rank decomposition in the algorithm, and we leave this for future research.  
Furthermore, the tensor train representation \cite{liu2020low} could be an alternative representation of the multidimensional array. 
In conclusion, our work provides an important direction for dealing with massive structured data as time-varying tensors for analysis in multi-modal neuroimaging studies. 

\section*{Acknowledgments}
The research of Dr. Lindquist is supported in part by NIH grants R01 EB016061 and R01 EB026549 from the National Institute of Biomedical Imaging and Bioengineering.
The research of Dr. Maiti is partially supported by the National Science Foundation grants NSF DMS-1952856 and 1924724.



\bibliographystyle{IEEEtran}
\bibliography{master-reference}

\appendices
\section*{Proof of Theorems 1 and 2}
In this section, we provide all technical details of materials described in Section 4. 
Our development is constructed upon and extended the previous work of  \cite{Lock2018,guhaniyogi2017bayesian,Xue2018} for different contexts and asymptotics are similar to reduced rank regression model (for example \cite{chen2013reduced}).
The lines of proof follow from \cite{newey1997convergence, bunea2011optimal, chen2013reduced}. 
\subsection{Technical lemmas}
\begin{lemma}
\label{lemma:trace}
For positive definite matrices $\bA$ and $\bB$ we have
\begin{equation}
    \lambda_{\min}(\bA) \trace\{\bB\}
    \leq 
    \trace\{\bA\bB\} \leq \lambda_{\max}(\bA)\trace\{\bB\}
\end{equation}
where $\lambda_{\max}(\bA)$ and $\bA$ and $\lambda_{\min}(\bA)$ are the largest eigenvalue and the smallest eigenvalues of $\bA$ respectively. 
\end{lemma}
\begin{proof}
See \cite{fang1994inequalities} for detailed proof. \hfill$\square$
\end{proof} 
\par
Before introducing the next lemma, let us define
a $P$-dimensional vector $\bu = (u_{1}, \cdots u_{P})^{\tp}$ which is sub-Gaussian with some parameters $\sigma$; then, for all $\balpha \in \bbR^{P}$,
\begin{equation}
    \E\{ \exp{\balpha^{\tp}\bu} \} \leq \exp(\|\balpha\|^{2}\sigma^{2}/2)
\end{equation}
Define the locally stationary time series $u_{j} = \sG(j/J, \sF_{j})$ where $\sF_{j} = (\cdots, \eta_{j-1}, \eta_{j}, \cdots)$; $\eta_{j}$s are i.i.d. random variables, and $\sG: [0,1]\times\bbR^{\infty} \rightarrow \bR$ is a measurable function such that $\xi_{j}(t) = \sG(t, \sF_{j})$.
Let $\{ \eta'\}$ be i.i.d. copies of $\eta$ and assume that for some $a > 0$, define the $L_{a}$-norm $\|\eta\|_{a} = \left\{\E|\eta|^{a}\right\}^{1/a}$. 
Then for $k \geq 0$ define the physical dependence measure $\Delta(k, a) = \sup_{t \in [0,1]}\max_{j}\|\sG(t, \sF_{j}) - \sG(t, \sF_{j, k})\|_{a}$
where $\sF_{j,k} = (\sF_{j-k-1}, \eta'_{j-k}, \eta_{j-k+1}, \cdots, \eta_{j})$. Moreover, recall the condition (A4) where for some large $a, \kappa_{0} > 0$, there exists a universal constant $C > 0$ such that $\Delta(k, a) \leq Ck^{-\kappa_{0}}$ for $k \geq 1$. 
Furthermore, let $\|\eta\|_{a}$ be finite for some $a > 1$. 
\begin{lemma}
\label{lemma:sigma}
Under condition (A4), and due to the above explanations, 
for some constant $C_{a} > 0$, 
\begin{equation}
    \bbP\left\{\frac{1}{NJ}\sigma_{1}(\sP\bE) \leq \frac{Q\xi N^{\tau}}{\sqrt{J}}\right\} \geq 1-C_{a}N^{-a\tau}
\end{equation}
where $\tau$ is some small positive real number and $\xi = \sup_{1\leq h\leq H}\sup_{t\in[0,1]}|\bbB_{h}(t)|$
\end{lemma}
\begin{proof}
See \cite{ding2021multivariate} and the references herein for the proof in detail. \hfill$\square$
\end{proof}

\begin{lemma}
\label{lemma:spline}
Define $\sS_{n}$ be a collection of spline such that the function $g_{\bullet}(t) = \sum_{h = 1}^{K_{N} + v + 1}b_{h, \bullet}B_{h}(t)$, 
where $\lbrace B_{h}, h = 1, \cdots,  (K_{N} + v + 1) \rbrace$ is a set of B-spline bases in $S_{n}$. Under conditions (A2) and (A3), there exists a spline function $g_{\bullet}(t) \in S_{n}$  such that 
\begin{equation}
    \sup_{t \in \sT} | \beta_{\bullet}(t) - g_{\bullet}(t)| = O\left(\frac{1}{K_{N}^{v + 1}}\right)
\end{equation}
\end{lemma}
\begin{proof}
This proof follows from \cite{de1978practical}. \hfill$\square$
\end{proof}


\subsection{Proof of Theorem 1}
For simplicity, assume $\bY \in \bbR^{NJ \times Q}$ 
and $\sfZ \in \bbR^{NJ \times H \times P}$, thus $\sfB \in \bbR^{H \times P \times Q}$. The contracted inner product in this proof is of order 2, i.e., $<\cdot, \cdot>_{2}$ , for simplicity, we drop subscript 2 from the inner product. 
By the definition of $\widehat{\sfB}_{0}$, for all matrices $\sfC$ of rank $R_{0}$ with order $H_{N}\times P \times Q$, we have 
\begin{align*}
    &\|\bY - \left< \sfZ, \widehat{\sfB}_{0} \right>\|_{\sF}^{2} + (NJ) \|\widehat{\sfB}_{0}\|_{\sF, \bW_{\omega}}^{2}\\
    &\leq 
    \|\bY - \left<\sfZ, \sfC\right>\|_{\sF}^{2} + (NJ) \|{\sfC}\|_{\sF, \bW_{\omega}}^{2}\numberthis
\end{align*}
In addition, the following two equations hold for any tensor $\sfC$,
\begin{align*}
\label{expanssion}
    &\|\bY-\left<\sfZ, \sfC \right>\|_{\sF}^{2} \\
    &=\|\bY-\left<\sfZ, \sfB_{0}\right>\|_{\sF}^{2} 
        + \|\left< \sfZ, (\sfB_{0} - \sfC)\right>\|_{\sF}^{2}\\
    & \qquad + 2\left<\bE, \left<\sfZ, (\sfB_{0}-\sfC)\right> \right>_{\sF}\\
    &\|\bY-\left<\sfZ, \widehat{\sfB}_{0}\right>\|_{\sF}^{2}\\
    &= \|\bY-\left<\sfZ, \sfB_{0} \right>\|_{\sF}^{2} 
        + \|\left< \sfZ, (\sfB_{0} - \widehat{\sfB}_{0}) \right>\|_{\sF}^{2} \\
    & \qquad + 2\left<\bE, \left<\sfZ, (\sfB_{0}-\widehat{\sfB}_{0}) \right> \right>_{\sF}\\
    \numberthis
\end{align*}
with 
$\left<\bA, \bB \right>_{\sF}= \trace\{\bA^{\tp}\bB\}$
for any matrices $\bA$ and $\bB$ such that the matrix product of $\bA^{\tp}\bB$ is permissible.
Define, $\sP = \sfZ_{(1)}(\sfZ_{(1)}^{\tp}\sfZ_{(1)})^{-1}\sfZ_{(1)}^{\tp}$, then by the definition of Frobenius inner product,  
$\left<\bE, \left< \sfZ, (\widehat{\sfB}_{0} - \sfB) \right> \right>_{\sF} = \left<\sP\bE, \left<\sfZ, (\widehat{\sfB}_{0} - \sfC)\right>\right>_{\sF}$.
Moreover, the inner product norm $\left<\cdot, \cdot \right>_{\sF}$, operator norm $\|\cdot\|_{2} = \sigma_{1}(\cdot)$ and nuclear norm $\|\cdot\|_{*} = \sum_{i}\sigma_{i}(\cdot)$ are related using the inequalities
$\left<\bA, \bB \right>_{\sF} \leq \|\bA\|_{2}\|\bB\|_{*}$ and $\|\bB\|_{*} \leq \sqrt{r}\|\bB\|_{\sF}$ where $r$ be the rank of the matrix $\bB$ and $\sigma_{i}(\cdot)$ represents the $i^\text{th}$ largest singular value of a matrix.
By subtracting the two Equations in (\ref{expanssion}) and 
exercising the properties of different norms mentioned above,
we get the following inequalities. 
\begingroup
\allowdisplaybreaks
\begin{align*}
\label{step1}
          &\|\left<\sfZ, (\widehat{\sfB}_{0} - \sfB_{0})\right>\|_{\sF}^{2}\\
          &\leq \|\left<\sfZ, (\sfC - \sfB_{0})\right>\|_{\sF}^{2} + 2\left<\bE, \left<\sfZ, (\widehat{\sfB}_{0}-\sfC)\right> \right>_{\sF} \\
          & \qquad + (NJ)
          \left\{
          \|\sfC\|^{2}_{\sF, \bW_{\omega}}
          -\|\widehat{\sfB}_{0}\|^{2}_{\sF, \bW_{\omega}}
          \right\}\\
          & = \|\left< \sfZ, (\sfC - \sfB_{0})\right>\|_{\sF}^{2} +
          2\left<\sP\bE, \left< \sfZ, (\widehat{\sfB}_{0}-\sfC)\right> \right>_{\sF}\\ 
          & \qquad + (NJ)\left\{
          \|\sfC\|^{2}_{\sF, \bW_{\omega}}
          -\|\widehat{\sfB}_{0}\|^{2}_{\sF, \bW_{\omega}}
          \right\}\\
          &\leq \|\left< \sfZ, (\sfC - \sfB_{0})\right>\|_{\sF}^{2} 
          + 2\sigma_{1}(\sP\bE)\sqrt{2R_{0}}\|\left<\sfZ, (\widehat{\sfB}_{0}- \sfC)\right>\|_{\sF}\\
          & \qquad + (NJ)\left\{
          \|\sfC\|^{2}_{\sF, \bW_{\omega}}
          -\|\widehat{\sfB}_{0}\|^{2}_{\sF, \bW_{\omega}}
          \right\}
          \numberthis
\end{align*}
\endgroup
Define, $\bP = \bI_{Q}\otimes\bI_{P}\otimes\int \bB''(t)\bB''(t)^{\tp}dt$ and observe the fact that $\lambda_{\max}(\bP) = \lambda_{\max}(\int \bB''(t)\bB''(t)^{\tp}dt)$. 
Now consider for any tensor with $\sfC$, using Lemma \ref{lemma:trace},
\begin{equation}
\begin{split}
    &\vect(\sfC)^{\tp}\bP\vect(\sfC) - \vect(\widehat{\sfB}_{0})^{\tp}\bP\vect(\widehat{\sfB}_{0})\\
    &= 
    \trace\{\bP(\vect(\sfC)\vect(\sfC)^{\tp} - \vect(\widehat{\sfB}_{0})\vect(\widehat{\sfB}_{0})^{\tp})\} \\
    & \leq \lambda_{\max}(\bP)\trace\{\vect(\sfC)\vect(\sfC)^{\tp} - \vect(\widehat{\sfB}_{0})\vect(\widehat{\sfB}_{0})^{\tp}\}\\
    & = \lambda_{\max}(\int \bB''(t)\bB''(t)^{\tp}dt)\{\|\sfC\|_{\sF}^{2} - \|\widehat{\sfB}_{0}\|_{\sF}^{2}\}
\end{split}
\end{equation}
As a consequence of the above inequality, 
\begin{equation}
    \begin{split}
    & \|\sfC\|_{\sF, \bW_{\omega}}^{2}
    - \|\widehat{\sfB}_{0}\|_{\sF, \bW_{\omega}}^{2}\\
    &= \vect(\sfC)^{\tp}\bW_{\omega}\vect(\sfC) -  \vect(\widehat{\sfB}_{0})^{\tp}{\bW}_{\omega}\vect(\widehat{\sfB}_{0})\\ 
    & = \theta\left(\vect(\sfC)^{\tp}\bP\vect(\sfC) \} 
    - \vect(\widehat{\sfB}_{0})^{\tp}\bP\vect(\widehat{\sfB}_{0}) \}\right)\\
    & \qquad + \phi\left(\vect(\sfC)^{\tp}\vect(\sfC)\}
    - \vect(\widehat{\sfB}_{0})^{\tp}\vect(\widehat{\sfB}_{0})\}
    \right)\\
    &\leq (\theta\lambda_{\max}(\bP) + \phi)
    \left\{\|\sfC\|_{\sF}^{2} -  \|\widehat{\sfB}_{0}\|_{\sF}^{2}\right\}\\
    & = (\theta\lambda_{\max}(\int \bB''(t)\bB''(t)^{\tp}dt) + \phi)
    \left\{\|\sfC\|_{\sF}^{2} -  \|\widehat{\sfB}_{0}\|_{\sF}^{2}\right\}
    \end{split}
\end{equation}
Then for tensor $\sfC$ with $\rank(\sfC) \leq R_{0}$ and $I = \min(H, PQ)$, we have the following inequalities 
\begingroup
\allowdisplaybreaks
\begin{align*}
     &\|\sfC\|_{\sF}^{2} -  \|\widehat{\sfB}_{0}\|_{\sF}^{2}\\
    &= \sum_{i = 1}^{I}\sigma_{i}^{2}(\sfC_{(1)}) - \sum_{i = 1}^{I}\sigma_{i}^{2}(\widehat{\sfB}_{0(1)})\\
    & \leq \left\{\sigma_{1}(\sfC_{(1)}) + \sigma_{1}(\widehat{\sfB}_{0(1)})\right\}
    \left\{\sum_{i = 1}^{I}\left(\sigma_{i}(\sfC_{(1)}) - \sigma_{i}(\widehat{\sfB}_{0(1)})\right)\right\}\\
    &\overset{(i)}{\leq} \left\{2\sigma_{1}(\sfC_{(1)}) + \sigma_{1}(\widehat{\sfB}_{0(1)} - \sfC_{(1)}) \right\}
    \left\{\sum_{i=1}^{I}\sigma_{i}(\widehat{\sfB}_{0(1)} - \sfC_{(1)})\right\}\\
    &= \left\{2\sigma_{1}(\sfC_{(1)}) + \sigma_{1}(\widehat{\sfB}_{0(1)} - \sfC_{(1)}) \right\}
    \left\{\sum_{i=1}^{R_{0}}\sigma_{i}(\widehat{\sfB}_{0(1)} - \sfC_{(1)})\right\}\\
    &\overset{(ii)}{\leq} \left\{2\sigma_{1}(\sfC_{(1)}) + \|\widehat{\sfB}_{0} - \sfC\|_{\sF}\right\}
    \left\{\sqrt{2R_{0}}\|\widehat{\sfB}_{0} - \sfC\|_{\sF}\right\}\\
    & \leq \sqrt{2R_{0}}
    \left\{ 
    2\sigma_{1}(\sfC_{(1)}) + \|\widehat{\sfB}_{0} - \sfC\|_{\sF}
    \right\}^{2}
    \numberthis
\end{align*}
\endgroup
where the inequality (i) follows since $\sigma_{i+j-1}(\bA+ \bB) \leq \sigma_{i}(\bA) + \sigma_{j}(\bB)$, or in other words due to Weyl additive perturbation theory which states that $\sigma_{i+j-1}(\bA) \leq \sigma_{i}(\bB) + \sigma_{j}(\bA - \bB)$. 
Inequality (ii) holds since by definition $\sigma_{1}(\bA) = \|\bA\|_{2}$, an operator norm; $\|\bA\|_{2} \leq \|\bA\|_{\sF}$; and due to Cauchy-Schwarz inequality along with the fact that $\rank(\bA + \bB) \leq rank(\bA) + rank(\bB)$. 
Also, 
$\|\left< \sfZ, (\widehat{\sfB}_{0} - \sfC) \right>\|_{\sF}^{2} = \| \left< \sfZ_{(1)}, (\widehat{\sfB}_{0} - \sfC)_{(3)} \right>\|_{\sF}^{2} = \| \sfZ_{(1)}^{\tp}(\widehat{\sfB}_{0} - \sfC)_{(3)}\|_{\sF}^{2}
\geq \|\widehat{\sfB}_{0} - \sfC\|_{\sF}^{2}\lambda_{\min}(\bZ_{(1)}^{\tp}\bZ_{(1)})$ due to Lemma \ref{lemma:trace}. Therefore, using the inequality 
$(x + y)^{2} \leq 2(x^{2} + y^{2})$
we have for $\mu = (NJ)(\theta\lambda_{\max}(\int \bB''(t)\bB''(t)^{\tp}dt) + \phi)\sqrt{2R_{0}}$
\begin{equation}
    \begin{split}
        &(NJ) \left\{\|\sfC\|_{\sF, \bW_{\omega}}^{2} - \|\widehat{\sfB}_{0}\|_{\sF, \bW_{\omega}}^{2}\right\}\\
        &\leq \mu
        \left\{
            2\sigma_{1}(\sfC_{(1)}) + \lambda^{-1}_{\min}(\sfZ_{(1)}^{\tp}\bZ_{(1)})\|\left<\sfZ, (\widehat{\sfB}_{0} - \sfC)\right>\|_{\sF}
        \right\}^{2}\\
        &\leq \mu
        \left\{
        4\sigma_{1}^{2}(\sfC_{(1)}) 
        + \lambda_{\min}^{-2}(\sfZ_{(1)}^{\tp}\sfZ_{(1)})\|\left< \sfZ, (\widehat{\sfB}_{0} - \sfC) \right>\|_{\sF}^{2} 
        \right\}\\
        & \leq 4\mu\sigma_{1}^{2}(\sfC_{(1)})
        + 2\mu\lambda_{\min}^{-2}(\sfZ_{(1)}^{\tp}\bZ_{(1)})
        \|\left< \sfZ, (\widehat{\sfB}_{0} - \sfB_{0})\right>\|_{\sF}^{2}\\
        & \qquad 
        + 2\mu\lambda_{\min}^{-2}(\sfZ_{(1)}^{\tp}\sfZ_{(1)})
        \|\left< \sfZ, (\sfC - \sfB_{0}) \right>\|_{\sF}^{2}
        \\
    \end{split}
\end{equation}
Therefore, we obtain the bound for the prediction error as the following way using the assumption that $\lambda_{\min}(\sfZ_{(1)}^{\tp}\sfZ_{(1)})$ is bounded below by $\lambda$ with high probability and by inequality $2xy \leq x^{2}/a + ay^{2}$ in $(\star)$, consider the following from Equation (\ref{step1}), 
\begingroup
\allowdisplaybreaks
\begin{align*}
        &\|\left< \sfZ, (\widehat{\sfB}_{0} - \sfB_{0}) \right>\|_{\sF}^{2}\\
        & \leq \|\left<\sfZ, (\sfC - \sfB_{0}) \right>\|_{\sF}^{2} 
        + 2\sigma_{1}(\sP\bE)\sqrt{2R_{0}}\|\left< \sfZ, (\widehat{\sfB}_{0} - \sfC) \right>\|_{\sF} \\
        & \qquad 
        + 2\mu\lambda^{-2}\|\left< \sfZ, (\widehat{\sfB}_{0} - \sfB_{0})\right>\|_{\sF}^{2}\\
        & \qquad + 2\mu\lambda^{-2}\|\left<\sfZ, (\sfC - \sfB_{0})\right>\|_{\sF}^{2} 
        + 4\mu\sigma_{1}^{2}(\sfC_{(1)})
        \\
        & \leq \|\left< \sfZ, (\sfC - \sfB_{0}) \right>\|_{\sF}^{2} 
        + 2\sigma_{1}(\sP\bE)\sqrt{2R_{0}}\|\left< \sfZ, (\widehat{\sfB}_{0} - \sfB_{0}) \right>\|_{\sF}\\
        & \qquad + 2\sigma_{1}(\sP\bE)\sqrt{2R_{0}}\|\left<\sfZ, (\sfC - \sfB_{0})\right>\|_{\sF}
        \\
        & \qquad + 2\mu\lambda^{-2}\|\left<\sfZ, (\widehat{\sfB}_{0} - \sfB_{0})\right>\|_{\sF}^{2} \\
        & \qquad +  2\mu\lambda^{-2}\|\left<\sfZ, (\sfC - \sfB_{0})\right>\|_{\sF}^{2} +4\mu\sigma^{2}_{1}(\sfC_{(1)})\\
        &\overset{(\star)}{\leq}
        4\mu\sigma_{1}^{2}(\sfC_{(1)}) + 
        \|\left< \sfZ, (\sfC - \sfB_{0}) \right>\|_{\sF}^{2}
        \\
        &\qquad+ 2R_{0}a\sigma_{1}^{2}(\sP\bE) + \|\left< \sfZ, (\widehat{\sfB}_{0} - \sfB_{0}) \right>\|_{\sF}^{2}/a \\
        & \qquad + 2\mu\lambda^{-2}\|\left<\sfZ, (\widehat{\sfB}_{0} - \sfB_{0})\right>\|_{\sF}^{2}
        \\
        &\qquad + 2R_{0}b\sigma_{1}^{2}(\sP\bE) + \|\left<\sfZ, (\sfC - \sfB_{0})\right>\|_{\sF}^{2}/b \\
        & \qquad +  2\mu\lambda^{-2}\|\left<\sfZ, (\sfC - \sfB_{0})\right>\|_{\sF}^{2}\\ 
        & \leq 
        4\mu\sigma_{1}^{2}(\sfC_{(1)})
        +2(a+b)R_{0}\sigma_{1}^{2}(\sP\bE)\\
        &\qquad +
        \left(
            \frac{b + 1}{b} + 2\mu\lambda^{-2}
        \right)
        \|\left< \sfZ, (\sfC - \sfB_{0}) \right>\|_{\sF}^{2} \\
        & \qquad + \left(
        \frac{1}{a} + 2\mu\lambda^{-2}
        \right)
        \|\left< \sfZ, (\widehat{\sfB}_{0} - \sfB_{0}) \right>\|_{\sF}^{2}
        \numberthis
\end{align*}
\endgroup
Therefore, by doing some algebra, we have, 
\begin{equation}
\begin{split}
    &\left(\frac{a-1}{a} - 2\mu\lambda^{-2}\right)\|\left< \sfZ, (\widehat{\sfB}_{0} - \sfB_{0}\right>\|_{\sF}^{2}\\
        & \leq 
        4\mu\sigma_{1}^{2}(\sfC_{(1)}) + 2(a+b)R_{0}\sigma_{1}^{2}(\sP\bE)\\
        &\qquad +
        \left(
        \frac{b+1}{b} + 2\mu\lambda^{-2}
        \right)
        \|\left< \sfZ, (\sfC - \sfB_{0})\right>\|_{\sF}^{2}
        \\
        &\|\left< \sfZ, (\widehat{\sfB}_{0} - \sfB_{0})\right>\|_{\sF}^{2}\\
        &\leq \left(
            \sC(\delta)^{-1} - 2\mu\lambda^{-2}
        \right)^{-1}
        \left\{
            4\mu\sigma_{1}^{2}(\sfC) + 2(1+\delta)
             R_{0}\sigma_{1}^{2}(\sP\bE)
        \right\}\\
        & + \left( 
            \frac{\sC(\delta) + 2\mu\lambda^{-2}}{\sC(\delta)^{-1} - 2\mu\lambda^{-2}}
        \right)\|\left< \sfZ, (\sfC - \sfB)\right>\|_{\sF}^{2}
\end{split}
\end{equation}
where $\sC(\delta) = 1 + 2/\delta$ and $\sigma_{1}(\sfC) = \max\{\sigma_{1}(\sfC_{(1)}), \sigma_{1}(\sfC_{(2)}), \sigma_{1}(\sfC_{(3)})\}$.
Last inequality holds after choosing $a = 1 + \delta/2$ and $b = \delta/2$. Now it is enough to provide an upper bound of the largest singular value of $\sP\bE$. 
For some positive constant $C_{0}$, with high probability $1-C_{0}N^{-a\tau}$, by lemma \ref{lemma:sigma},
\begin{align*}
    &\|\left<\sfZ, (\widehat{\sfB}_{0} - \sfB_{0})\right>\|_{\sF}^{2}\\
        &\leq \left(
            \sC(\delta)^{-1} - 2\mu\lambda^{-2}
        \right)^{-1}
        \left\{
            4\mu\sigma_{1}^{2}(\sfC) + \right.\\
            &\left.
            {2R_{0}}(1+\delta)Q^{2}\xi^{2} N^{2\tau+2}{J}
        \right\}\\
        & + \left( 
            \frac{\sC(\delta) + 2\mu\lambda^{-2}}{\sC(\delta)^{-1} - 2\mu\lambda^{-2}}
        \right)\|\left< \sfZ, (\sfC - \sfB_{0}) \right>\|_{\sF}^{2}
        \numberthis
\end{align*}
Since $\sfC$ is an arbitrary matrix with $\rank(\sfB) \leq R_{0}$, the choosing $\sfC = \sfB_{0}$, we have, 
\begin{align*}
    \label{eq:predictionError}
     &\|\left<\sfZ, (\widehat{\sfB}_{0} - \sfB_{0}) \right>\|_{\sF}^{2}\\
     &\leq \left(
            \sC(\delta)^{-1} - 2\mu\lambda^{-2}
        \right)^{-1}
            \left\{
            4\mu\sigma_{1}^{2}(\sfC) + {2R_{0}}(1+\delta)Q^{2}\xi^{2} N^{2\tau+2}{J}
        \right\}
\end{align*}
\par
Estimation bound can be derived from the above expression under condition $\lambda_{\min}(\sfZ_{(1)}^{\tp}\sfZ_{(1)}) \geq \lambda$, from inequality \ref{eq:predictionError}, we have, 
\begin{align*}
    &\|\widehat{\sfB}_{0} - \sfB_{0}\|_{\sF}^{2}\\
    &\leq \lambda^{-1}\left(
            \sC(\delta)^{-1} - 2\mu\lambda^{-2}
        \right)^{-1}\\
    & \qquad \times    
        \left\{
            4\mu\sigma_{1}^{2}(\sfC) + {2R_{0}}(1+\delta)Q^{2}\xi^{2} N^{2\tau+2}{J}
        \right\}
        \numberthis
\end{align*}

\subsection{Proof of Theorem 2}
Observe that, due to Lemma \ref{lemma:spline} and the fact that, 
\begin{align*}
    &\int_{\sT}[\bbB(t)^{\tp}(\widehat{\sfB}_{0} - \sfB_{0})]^{2}f_{T}(t)dt\\
    &= 
    (\widehat{\sfB}_{0} - \sfB_{0})^{\tp}
    \left(\int_{\sT}\bbB_{h}(t)\bbB_{h}(t)^{\tp}f_{T}(t)dt\right)
    (\widehat{\sfB}_{0} - \sfB_{0})\\
    & \propto 
    \|\widehat{\bB}_{0} - \bB_{0}\|^{2} = O_{P}(a_{N})
    \numberthis
\end{align*}
So, we can derive
\begin{equation}
    \begin{split}
        & \int_{\sT}(\widehat{\beta}_{\bullet}(t) - \beta_{\bullet}(t))^{2}f_{T}(t)dt\\
        &=
        \int_{\sT}\left(
        \bbB(t)^{\tp}(\widehat{\sfB}_{0} - \sfB_{0}) + \bbB(t)^{\tp}\sfB_{0} - \beta(t) \right)^{2}f_{T}(t)dt\\
        &\leq 2\|\widehat{\sfB}_{0} - \sfB_{0}\|_{\sF}^{2} 
        + 2\int_{\sT}[\bbB(t)^{\tp}\sfB_{0} - \beta(t)]^{2}f_{T}(t)dt\\
        & \leq  O_{P}(a_{N}) + O(K_{N}^{-2(\nu + 1)})
\end{split}
\end{equation}

\section*{More details of \texttt{forrestGump} data in Section \ref{sec:real-data}}
In the audio-visual movie, the video track of the movie was extracted and encoded as H.264 ($1280 \times 720$ at 25 fps). 
The movie was shown on a $1280 \times 1024$ pixel screen with a 63 cm viewing distance in 720p resolution. The temporal resolution of the participants' eye gaze recording was 1000Hz.
\par
All fMRI acquisitions had the following parameters: T2*- weighted echo-planner images with 2 sec repetition time (TR), 30 ms echo time, and 90-degree flip angle were acquired during stimulation using a 3 Tesla MRI scanner. The dimension of the images for each time-point was $80 \times 80 \times 35$ (with pixel dimension $3 \times 3 \times 3.3 mm^{3}$). The number of volumes acquired for the selected session was 451.
\par
Brain imaging data comes directly from the scanner and hence it is difficult to answer scientific questions based on these raw data. 
Therefore, pre-processing of fMRI data plays an important role for studying imaging data. 
Pre-processing steps are performed by \texttt{fslr} package in R \cite{jenkinson2012fsl,muschelli2015fslr}. 
Slice timing correction method corrects the variability in the BOLD responses that are due to the fact that data in different voxels are acquired at different time. 
This step has been performed using the function \textit{slicetimer} whether indexing is done from top and order of the acquisition is continuous. 
Later \textit{bias\_correct} function is used for bias field corrections. 
After that, motion correction is performed to correct the variability due to head movement. 
Motion correction is a special case of image registration where a series of images are aligned by considering mean image over all time-points as target image for each individuals. 
It is easy to visualise that any rigid body movement can be described by six parameters. When a subject lies inside the scanned, the center of any voxel is its head occupies a point in space that can be characterised by triplet (x, y, z). 
By convention, z-axis is parallel to the bore of the magnet and x-axis is passing through the subject ears from left to right side and y-axis is a pole that enters through the back of the head and exits in forehead. 
Based on this coordinate system, possible rigid body movements are translation along x, y and z axes and rotation about x, y and z axes. Mean BOLD responses is taken as the standard and then rigid body transformation is performed for rest of TRs until each of the data sets agrees as closely as possible with the mean data. 
Motion corrected images have same dimension, voxel spacing, origin and direction as the images gathered from scanner.
Here we use \textit{antsrMotionCalculation} function which provides an R-wrapper around the Insight Segmentation and Registration Toolkit (ITK). A calculated frame-wise motion parameters could be obtained due to a rigid body transformation that was performed which can be described by six parameters as illustrated in Figure \ref{fig:appendix-mocoParams-task} where three parameters contain the rotation matrix (rotation along x, y and z axes respectively) and other three parameters are translation vectors (translation along x, y and z axes respectively) at each TR. Additionally motion-corrected time-series data has be provided. 
\par
The goal of the next step is to align the functional and structural images to improve the spatial resolution. 
Brain activity is restricted to brain tissue only, therefore brain extraction of the anatomical image must be performed to remove artifacts.
Furthermore, the functional brain atlas provides information on the location of the functional brain region, aggregating knowledge on the brain functionality. 
Here we use an atlas proposed by Montreal Neurological Institute (MNI) where MNI-atlas was created by averaging the results from high resolution structural images taken over 152 different brains with dimension $182 \times 218 \times 182$ with pixel dimension $1\times1\times1 \text{mm}^{3}$ and it is also provided in FSL as MNI152\_T1\_1mm\_brain. 
Spatial smoothing to the data to reduce non-systematic high frequency spatial noise is conducted which subsequently reduces high-frequency noise that changes quickly across small regions of the brain, we take (6, 6, 7) as kernel width (FWHM). Temporal filtering is used to reduce the effect of slow fluctuations in the local magnetic field properties to the scanner. 
Interested readers are encouraged to study \cite{ashby2011statistical,wager2015principles} for more details about the pre-processing steps. 

\begin{figure*}[b]
       \includegraphics[width=\textwidth]{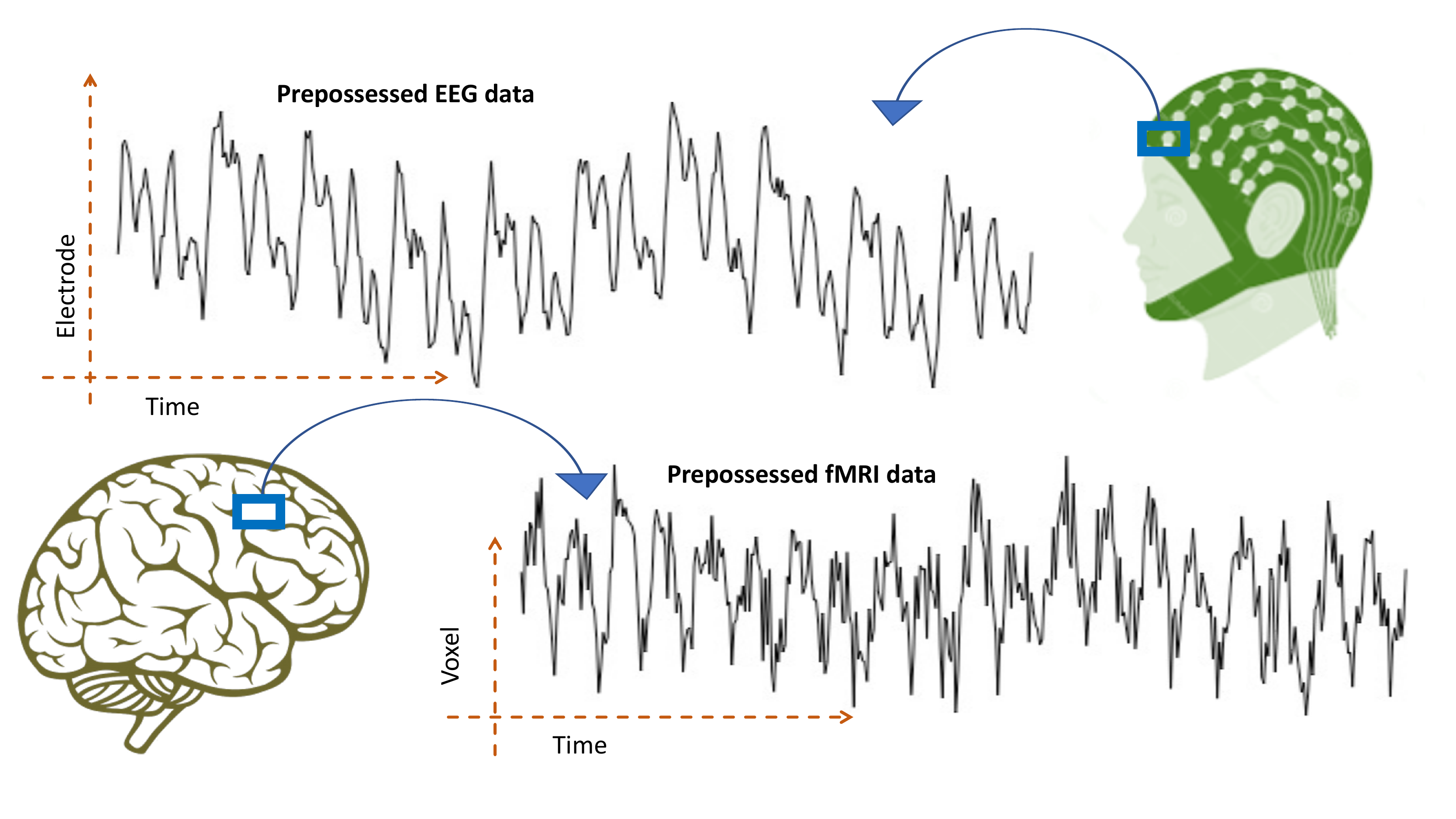}
       \caption{\texttt{Multi-modal data (top part)}: An example of multi-modal data analysis which seeks to explore the relationship between EEG and fMRI data.}
       \label{fig:multimodal}
\end{figure*}

\begin{figure*}[b]
 \subfloat{%
    \includegraphics[width=\textwidth]{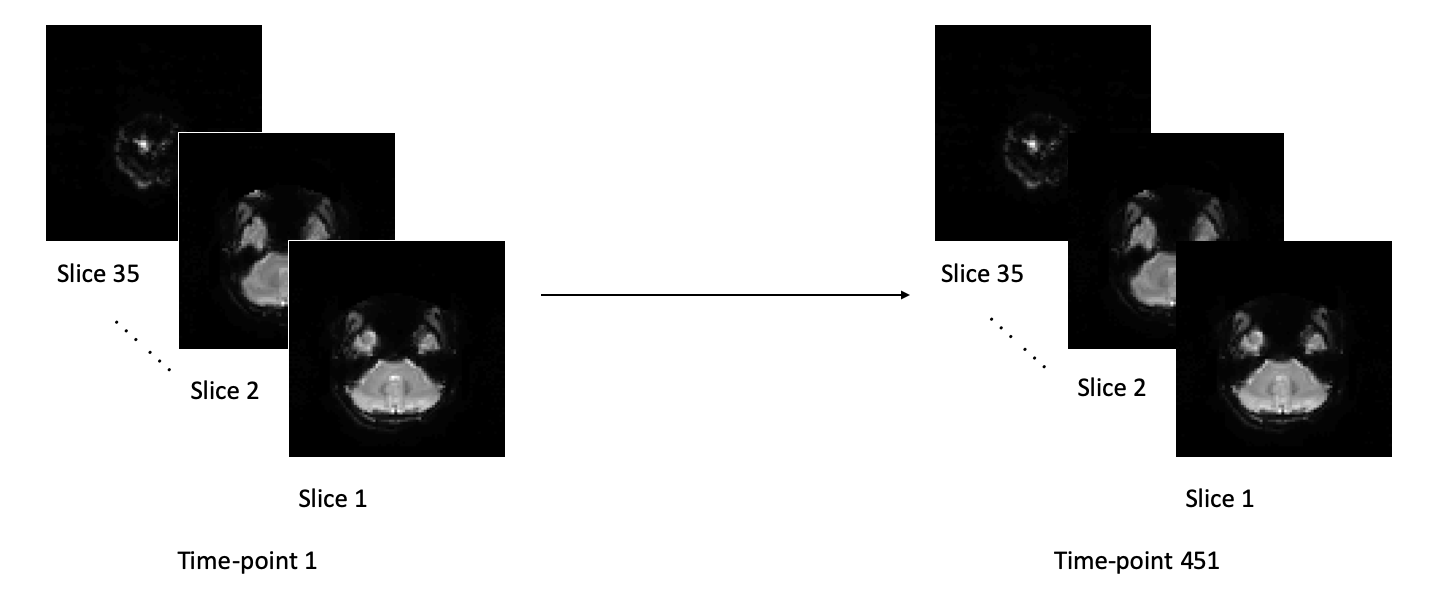}
}\\
\subfloat{%
    \includegraphics[width=\textwidth]{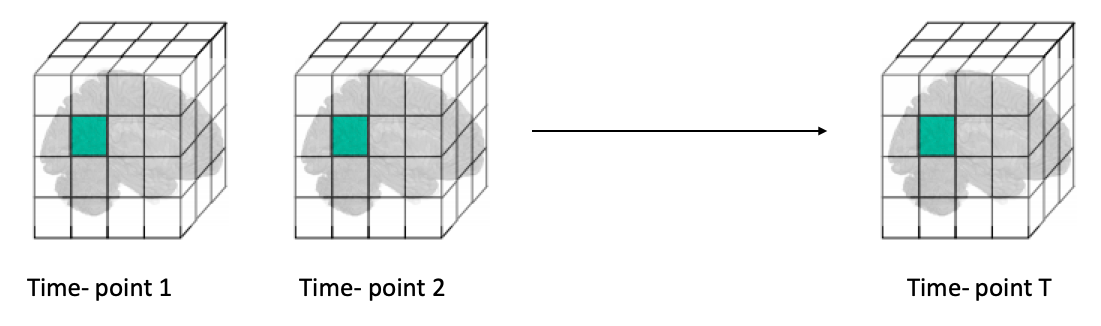}
 }
 \caption{
\texttt{ForrestGump data (bottom part)}: 
 (Top panel) BOLD fMRI for an example subject during their first run (see Section \ref{sec:real-data} for details). 35 axial slices (thickness 3.0 mm) represents the third mode of the tensor with $80 \times 80$ voxels ($3.0 \times 3.0$ mm) in-plate resolution measured at every repetition time (TR) of 2 seconds. 
 (Bottom panel) fMRI dataset consists of a time series of 3D images (tensors) at each TR (source: \cite{wager2015principles}).
 }\label{fig:whyTensor}
\end{figure*}

\begin{figure*}[b]
    \centering
    \includegraphics[width=0.45\textwidth]{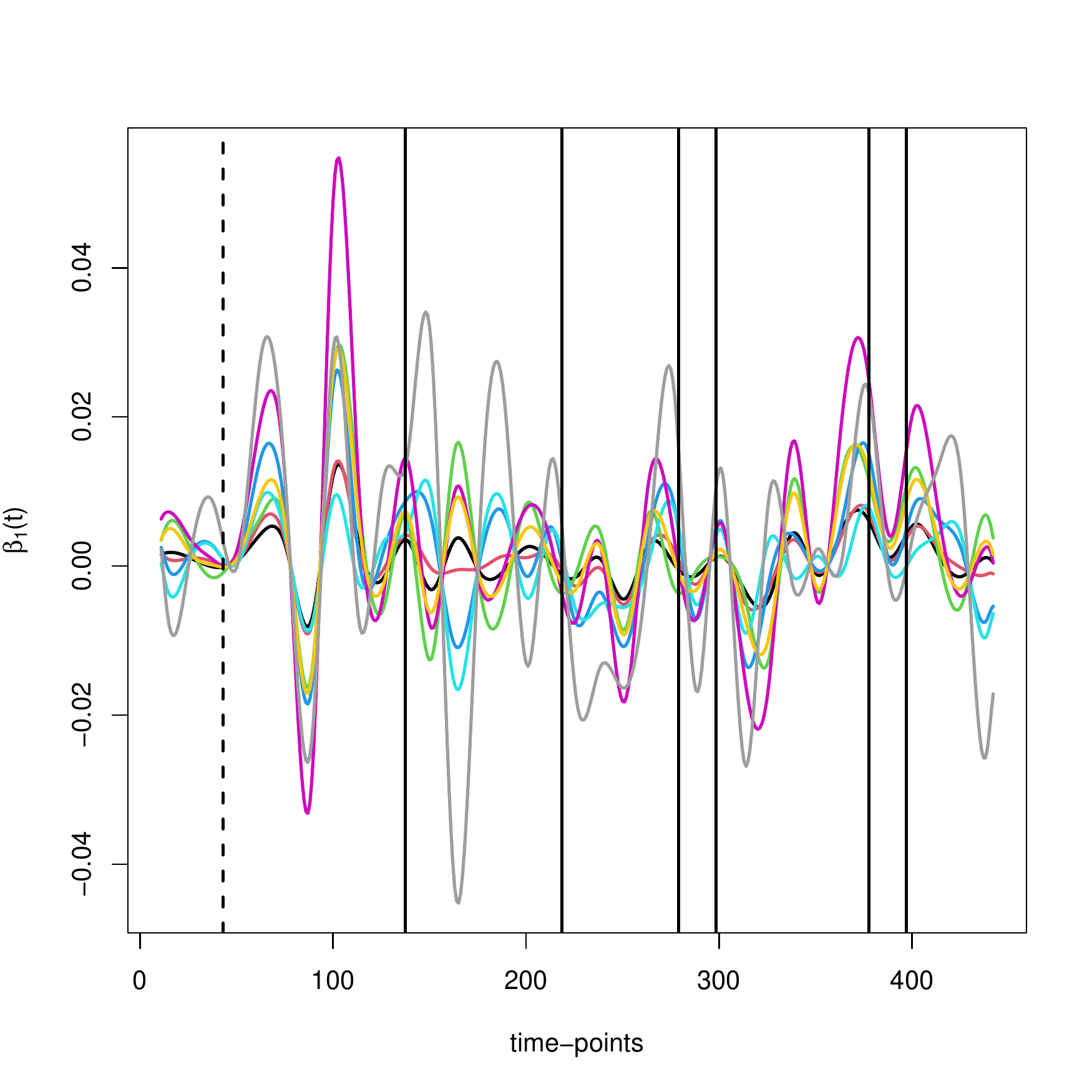}
    \includegraphics[width=0.45\textwidth]{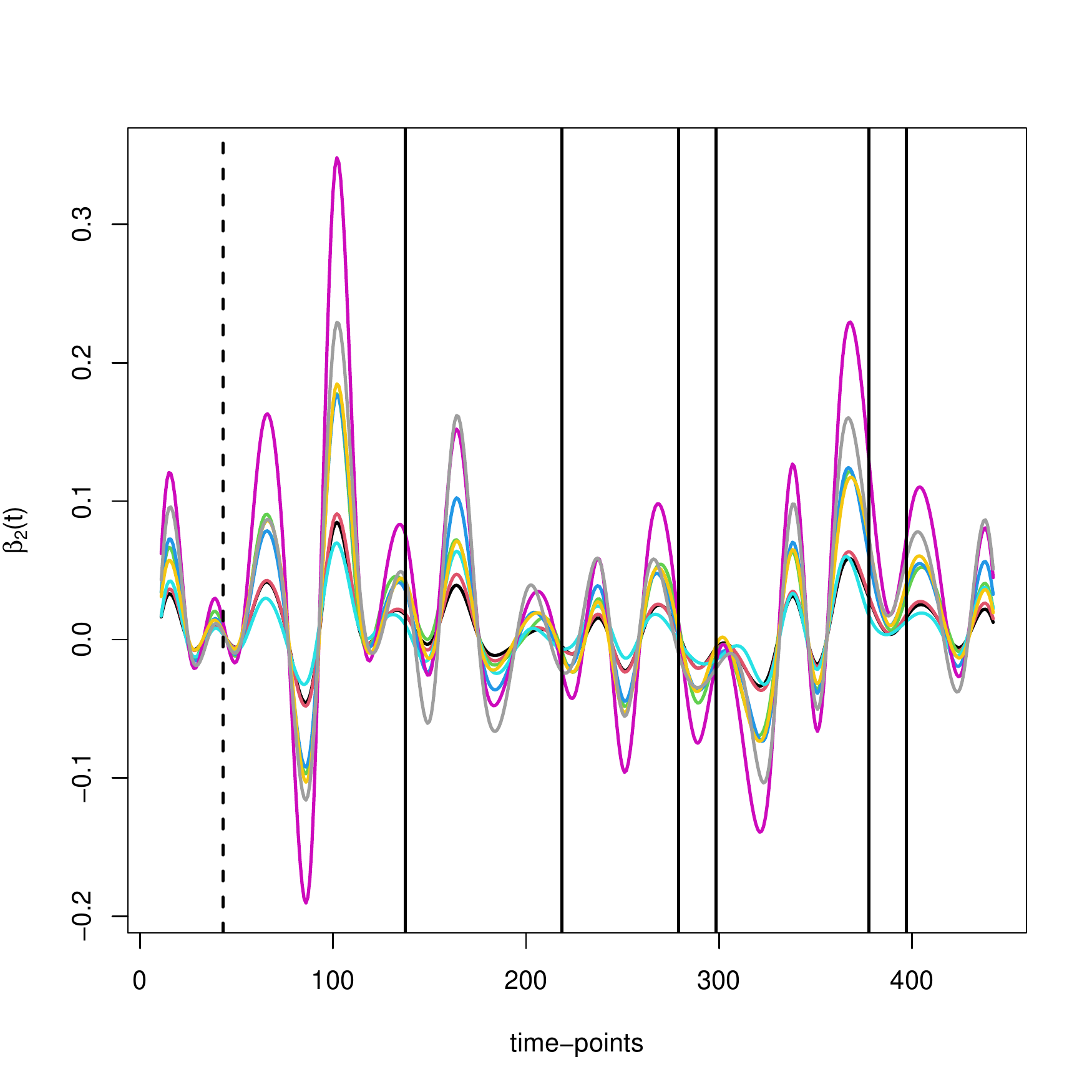}
    \includegraphics[width=0.45\textwidth]{tensor-beta2-rank3ext2int.pdf}
    \includegraphics[width=0.45\textwidth]{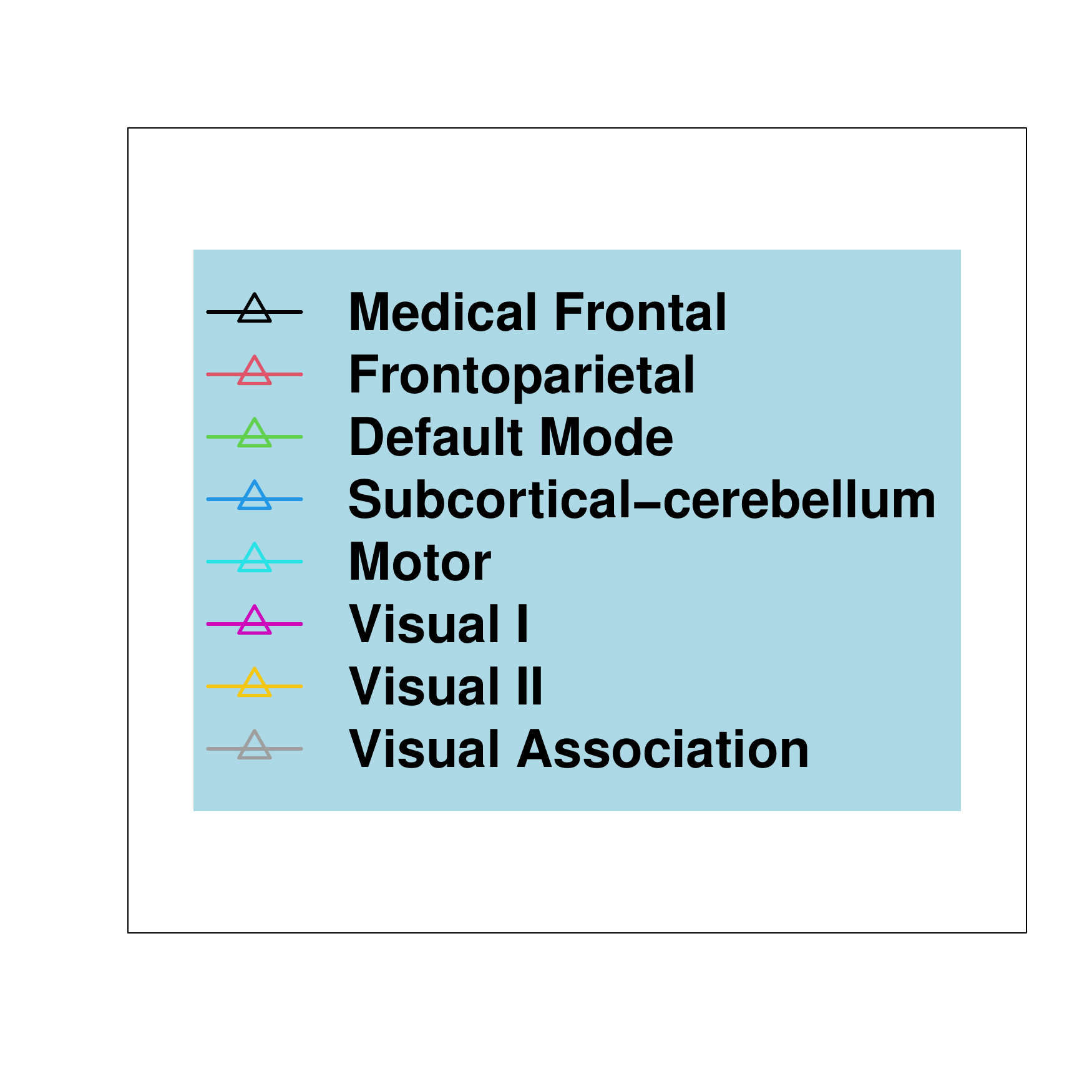}
    \caption{\texttt{forrestGump-data} results: Estimates of the coefficients corresponding to three visual features. 
    Here each panel represents estimated $\bbeta(t)$ corresponding to
    distance, angle of eye-gaze and pupil area, respectively,
    over eight different brain networks. 
    The dashed vertical lines represent scene changes in the movie. The first dashed line depicts the end of the opening credits. Subsequent bold lines represent scene changes that alternate between interior and exterior settings.}
    \label{fig:coef}
\end{figure*}

\begin{figure*}[b]
    \centering
    \includegraphics[scale=0.35]{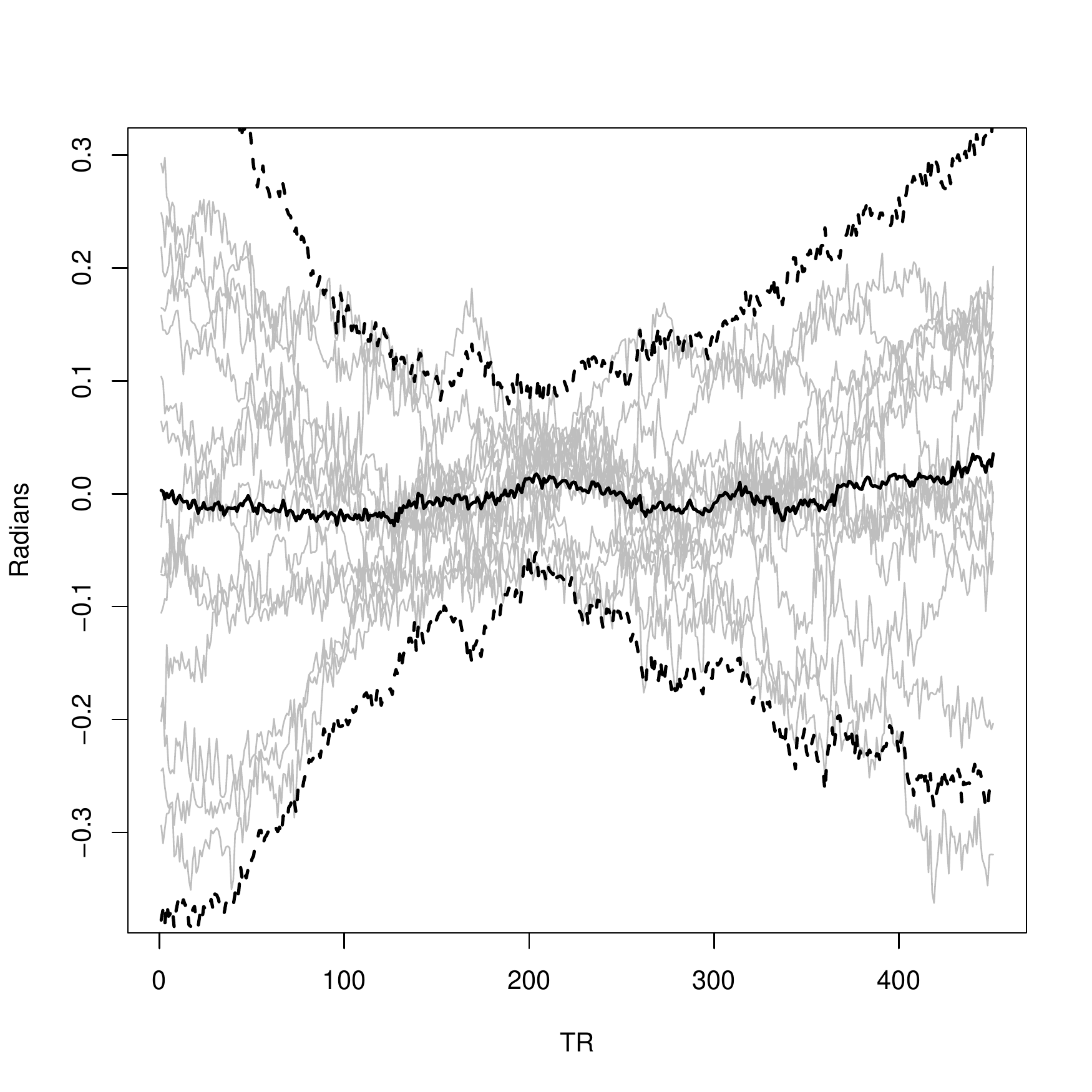}
    \includegraphics[scale=0.35]{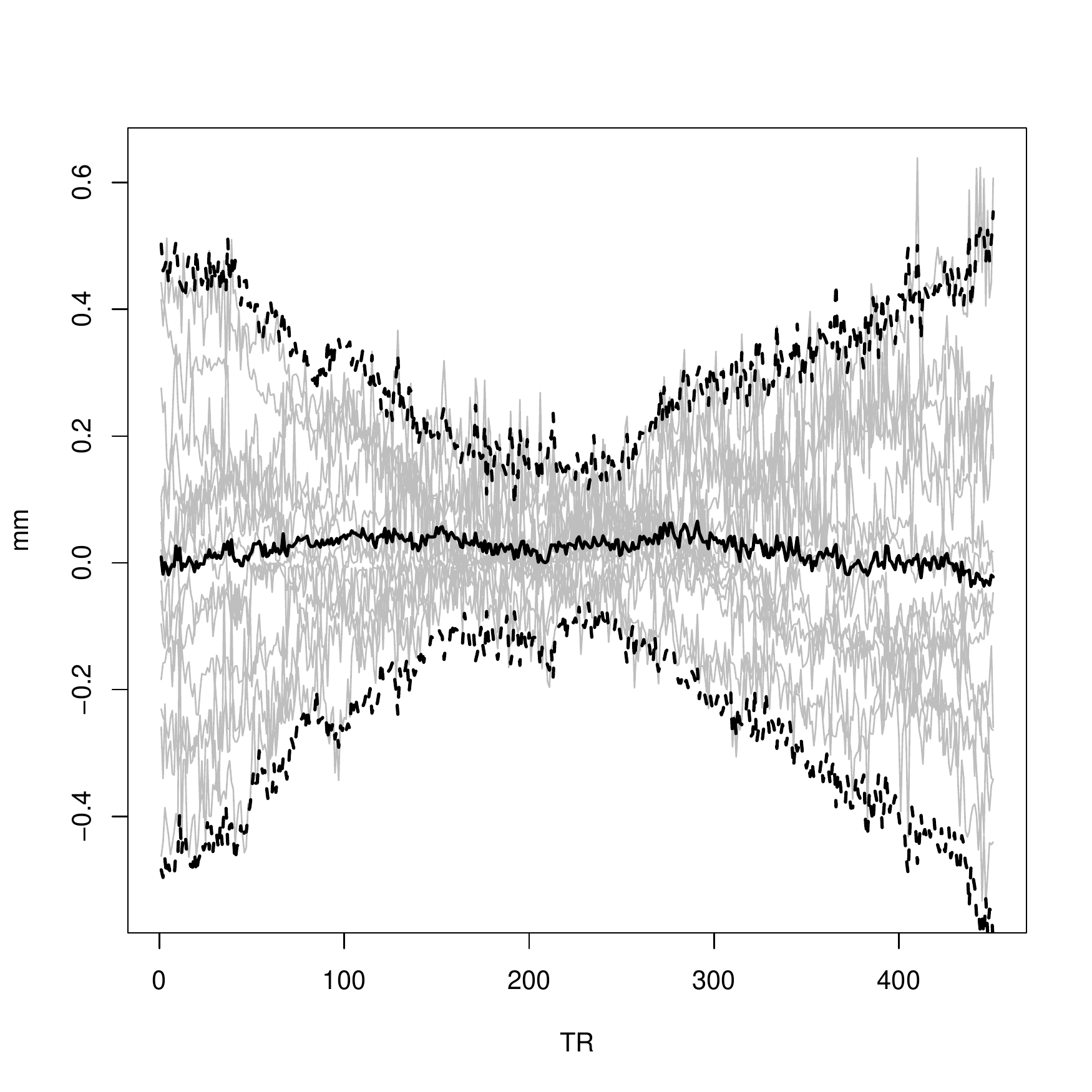}
    \includegraphics[scale=0.35]{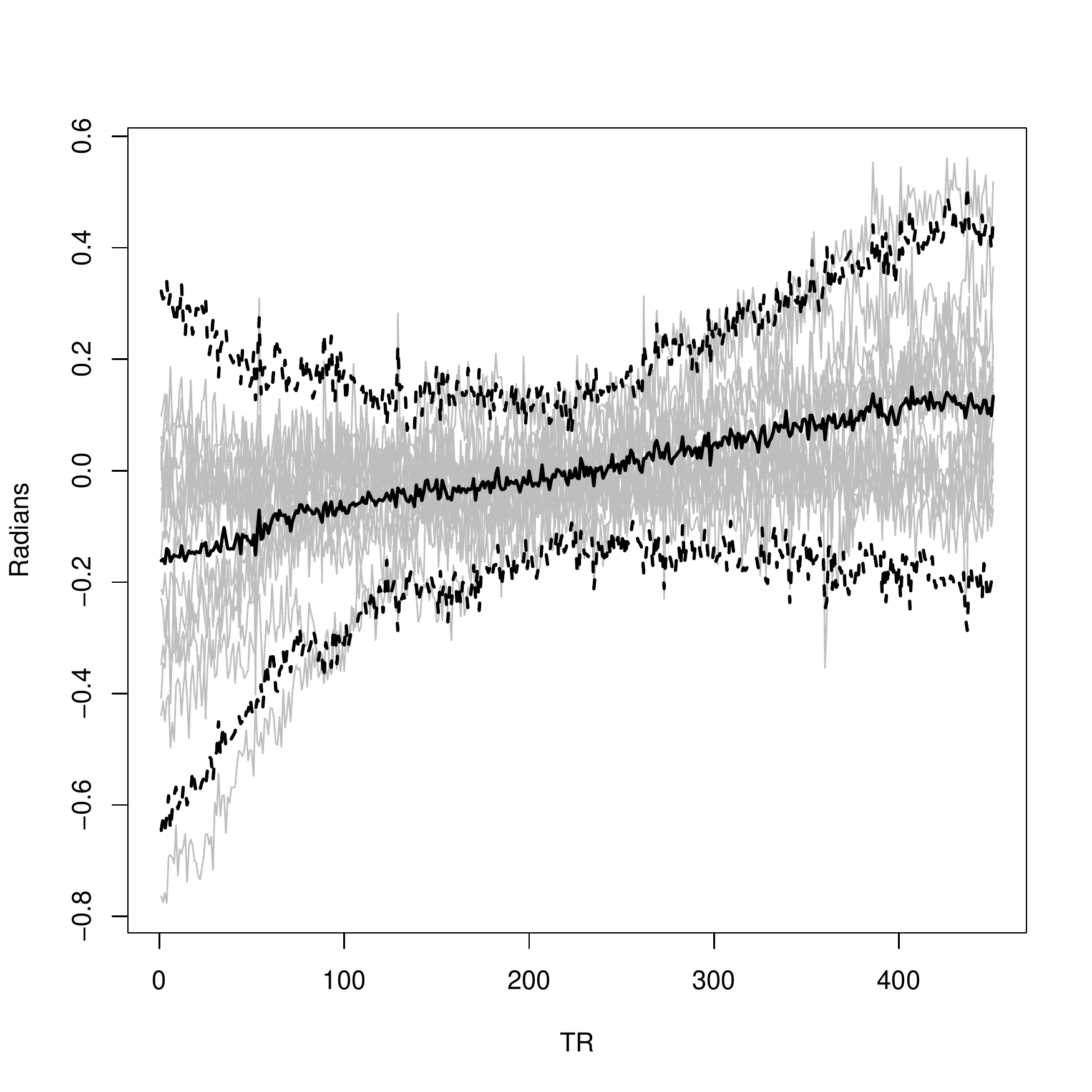}
    \includegraphics[scale=0.35]{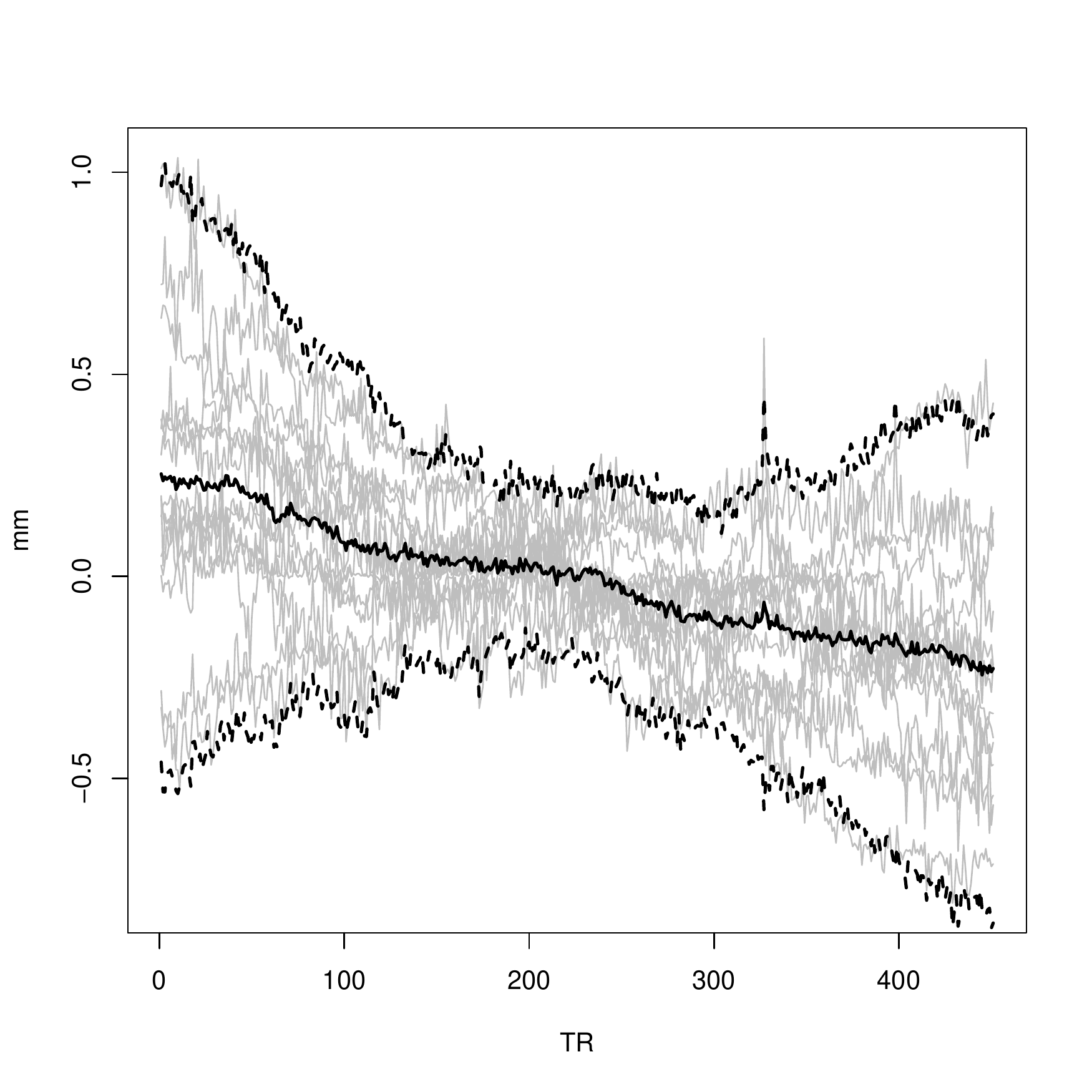}
    \includegraphics[scale=0.35]{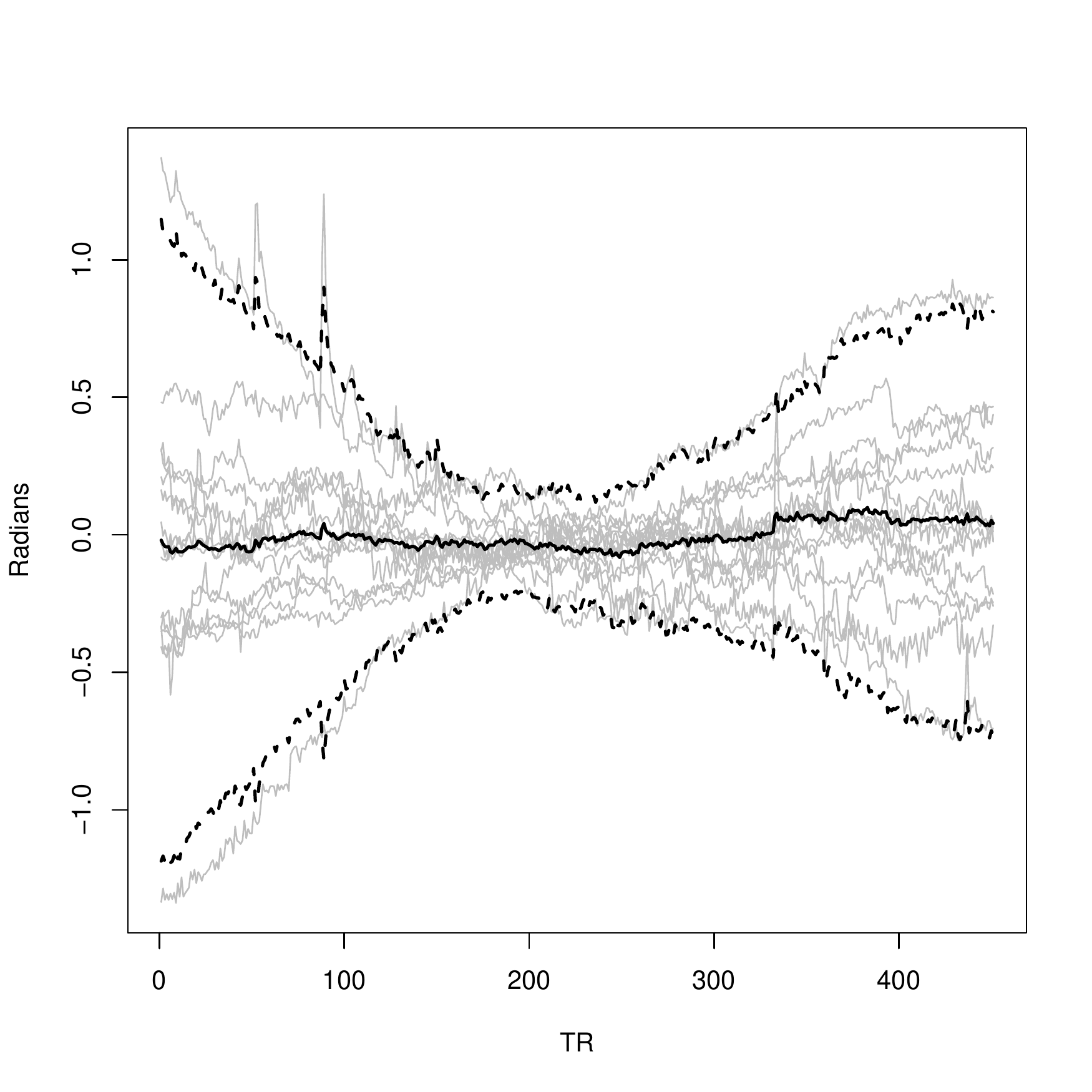}
    \includegraphics[scale=0.35]{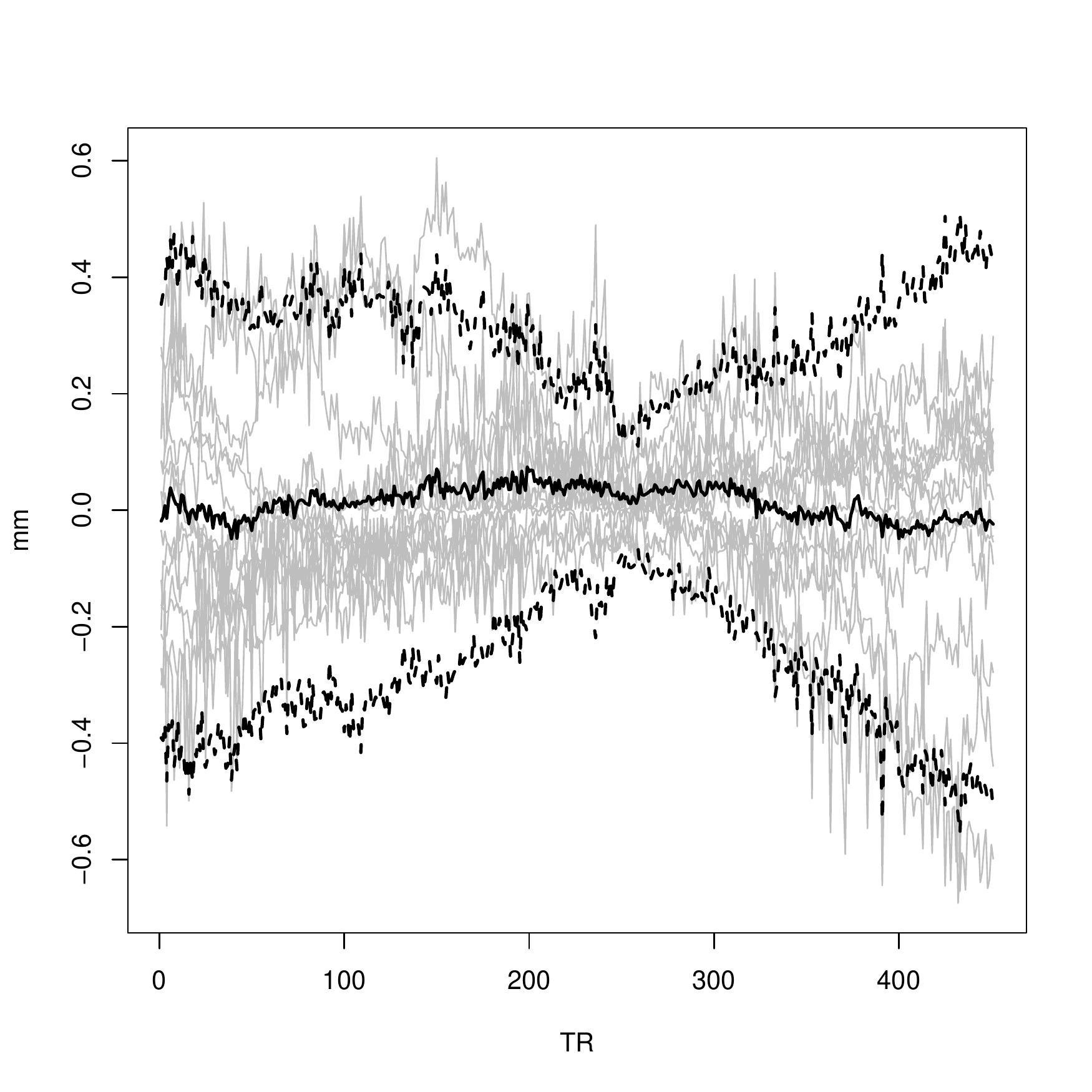}
    \caption{\texttt{data-forrestGump}: Summary statistics for the parameters estimates of head motion correction across TRs and participants. The left panel shows the magnitude of the three rotational parameters (in radians) for each individual on each of the 451 TRs and right panel shows the magnitude of the three translation parameters (in millimeters) for each individual on each of the 451 TRs. In each plot, solid black line indicates the mean over the individuals through TRs and black dotted lines indicate the mean$\pm$2sd over the individuals through TRs. 
    }
    \label{fig:appendix-mocoParams-task}
\end{figure*}




\vfill

\end{document}